\documentclass[11pt]{article}

\usepackage[utf8]{inputenc}
\usepackage[T1]{fontenc}

\usepackage{fullpage}
\usepackage{amsmath,amsfonts,amsthm,xspace,hyperref,graphicx,relsize,bm}
\usepackage{mathpazo}
\usepackage[margin=1in]{geometry}
\usepackage{float}
\usepackage{endnotes}
\usepackage{enumitem}
\usepackage{xcolor}
\usepackage{cleveref}
\usepackage{stmaryrd}
\usepackage{dsfont}
\usepackage{thm-restate}

\usepackage{url}
\usepackage{qcircuit}
\usepackage{upgreek}

\definecolor{DarkBlue}{rgb}{0,0,0.8}  %
\definecolor{DarkOrange}{rgb}{0.8,0.4,0}  %
\def\mylinkcolor{DarkBlue}

\hypersetup{%
	colorlinks=true, urlcolor=\mylinkcolor,
  linkcolor=\mylinkcolor, citecolor=\mylinkcolor}

\allowdisplaybreaks

\usepackage{mathtools}

\newtheorem{theorem}{Theorem}[section]
\newtheorem{proposition}[theorem]{Proposition}
\newtheorem{lemma}[theorem]{Lemma}
\newtheorem{claim}[theorem]{Claim}

\newtheorem{definition}[theorem]{Definition}
\theoremstyle{definition}

\newtheorem{conjecture}[theorem]{Conjecture}

\numberwithin{equation}{section}

\usepackage[normalem]{ulem}
\usepackage{todonotes}
\newcommand{\complex}{{\mathbb C}}
\newcommand{\reals}{{\mathbb R}}

\newcommand{\integers}{{\mathbb Z}}

\newcommand{\tensor}{\otimes}

\newcommand{\adjoint}{*}
\newcommand{\eqdef}{\coloneqq}

\newcommand{\ket}[1]{| #1 \rangle}
\newcommand{\bra}[1]{ \langle #1 |}
\newcommand{\ketbra}[2]{| #1 \rangle\!\langle #2 |}
\newcommand{\braket}[2]{\langle #1 | #2 \rangle }
\newcommand{\density}[1]{\ketbra{#1}{#1}}

\newcommand{\norm}[1]{\left\| #1 \right\|}
\newcommand{\trnorm}[1]{\left\| #1 \right\|_{\mathrm{tr}}}
\newcommand{\sg}{{\uppsi_2}}

\newcommand{\size}[1]{\left| #1 \right|}
\newcommand{\set}[1]{\left\{ #1 \right\}}

\newcommand{\transpose}{\top}

\newcommand{\e}{{\mathrm e}}
\newcommand{\complexi}{{\mathrm{i}}}

\newcommand{\id}{{\mathbb 1}}

\newcommand*\diff{\mathop{}\!\mathrm{d}}

\DeclareMathOperator{\expct}{{\mathbb E}}
\newcommand{\Order}{\mathrm{O}}
\newcommand{\order}{\mathrm{o}}

\newcommand{\linear}{{\mathsf L}}
\newcommand{\unitary}{{\mathsf U}}

\newcommand{\qstate}{{\mathsf D}}

\DeclareMathOperator{\trace}{Tr}
\DeclareMathOperator{\support}{supp}

\DeclareMathOperator{\entropy}{H}
\DeclareMathOperator{\mi}{I}

\DeclareMathOperator{\Tr}{Tr}
\newcommand{\Id}{\id}
\renewcommand{\id}{\mathds{1}}

\DeclareMathOperator*{\Ex}{\mathbb{E}}

\renewcommand{\cal}[1]{\mathcal{#1}}

\newcommand{\Ind}{\mathbf{1}}
\newcommand{\C}{\ensuremath{\mathbb{C}}}

\newcommand{\R}{\ensuremath{\mathbb{R}}}

\newcommand{\eps}{\varepsilon}

\newcommand{\Paren}[1]{\left(#1\right)}

\newcommand{\Abs}[1]{\left\lvert#1\right\rvert}

\newcommand{\PauliX}{{\mathrm X}}
\newcommand{\PauliY}{{\mathrm Y}}
\newcommand{\PauliZ}{{\mathrm Z}}
\newcommand{\Swap}{{\mathrm F}}
\newcommand{\mes}{\upphi} %

\newcommand{\bpsi}{{\bm \psi}}

\newcommand{\bU}{{\bm U}}
\newcommand{\bj}{{\bm j}}
\newcommand{\bk}{{\bm k}}
\newcommand{\bA}{{\bm A}}
\newcommand{\bM}{{\bm M}}

\newcommand{\bmu}{{\bm \mu}}
\newcommand{\bQ}{{\bm Q}}
\newcommand{\bR}{{\bm R}}
\newcommand{\bS}{{\bm S}}
\newcommand{\bL}{{\bm \Lambda}}
\newcommand{\bx}{{\bm x}}
\newcommand{\bv}{{\bm v}}

\newcommand{\bs}{{\bm s}}
\newcommand{\bzeta}{{\bm \zeta}}
\newcommand{\bxi}{{\bm \xi}}

\newcommand{\bPi}{{\bm \Pi}}
\newcommand{\bcE}{{\bm {\mathcal E}}}
\newcommand{\bE}{{\bm E}}
\newcommand{\beps}{{\bm \eps}}

\DeclareMathOperator{\diracdelta}{\updelta}    %
\DeclareMathOperator{\dist}{\upgamma}    %
\DeclareMathOperator{\hc}{\upalpha_{\text{Holevo}}}    %

\newcommand{\rd}{{\mathrm d}}
\newcommand{\rP}{{\mathrm P}}
\newcommand{\EPR}{\mathrm{EPR}}

\newcommand{\nhs}{\mathrm{nhs}}

\newcommand{\cF}{{\mathcal F}}
\DeclareMathOperator{\MP}{p}

\newcommand{\suppress}[1]{}
\newcommand{\comment}[1]{}

\begin{document}

\title{Rigidity of superdense coding}

\author{Ashwin Nayak~\thanks{Department of Combinatorics and Optimization,
and Institute for Quantum Computing, 
University of Waterloo, 200 University Ave.\ W., Waterloo, ON,
N2L~3G1, Canada. Email: \texttt{ashwin.nayak@uwaterloo.ca}~.} 
\and 
Henry Yuen~\thanks{Department of Computer Science, Columbia University, New York, USA. E-mail: \texttt{henry.yuen@columbia.edu}~.}
}

\date{}
\maketitle

\begin{abstract}
The famous superdense coding protocol of Bennett and Wiesner demonstrates that it is possible to communicate two bits of classical information by sending only one qubit and using a shared EPR pair. Our first result is that an arbitrary protocol for achieving this task (where there are no assumptions on the sender's encoding operations or the dimension of the shared entangled state) is locally equivalent to the canonical Bennett-Wiesner protocol. In other words, the superdense coding task is \emph{rigid\/}. In particular, we show that the sender and receiver only use additional entanglement (beyond the EPR pair) as a source of classical randomness.

We also investigate several questions about higher-dimensional superdense coding, where the goal is to communicate one of $d^2$ possible messages by sending a $d$-dimensional quantum state, for general dimensions~$d$. Unlike the $d=2$ case (i.e. sending a single qubit), there can be inequivalent superdense coding protocols for higher~$d$. We present concrete constructions of inequivalent protocols, based on constructions of inequivalent orthogonal unitary bases for all $d > 2$. 
Finally, we analyze the performance of superdense coding protocols where the encoding operators are independently sampled from the Haar measure on the unitary group. Our analysis involves bounding the distinguishability of random maximally entangled states, which may be of independent interest.

\end{abstract}

\tableofcontents

\newcommand{\DivMid}{\, \Big \| \,}
\newcommand{\Hilb}{\mathcal{H}}
\newcommand{\wt}[1]{\widetilde{#1}}

\section{Introduction}

In quantum information theory, rigidity is a phenomenon where optimal performance in an information processing task requires using a protocol satisfying extremely stringent constraints --- in some cases, there is essentially a \emph{unique\/} optimal protocol. The primary examples of rigidity come from nonlocal games (also known as \emph{Bell tests} in the physics literature). In this setting two spatially separated parties Alice and Bob play a game with a third-party called the referee. In order to maximize their chances of winning, before the game starts Alice and Bob choose an entangled state to share as well as local measurements to perform on the state. For example, in the famous CHSH game the optimal winning probability is $\cos^2(\pi/8)$, and a canonical strategy that achieves this uses a (rotated) EPR pair and single-qubit Pauli measurements. The CHSH game is \emph{rigid\/} in the sense that \emph{any\/} optimal strategy for the CHSH game is identical to this canonical strategy, up to local changes of basis. 

The study of rigidity in quantum information processing arguably started with the work of Mayers and Yao~\cite{mayers1998quantum,mayers2003self}, who initiated the concept of \emph{device-independent cryptography}. The idea behind this subject is that a classical user can verify that untrusted quantum hardware is behaving as intended --- say, generating random keys or performing a quantum computation -- simply by verifying that the hardware is employing a (near)-optimal strategy in a rigid nonlocal game. Since the work of Mayers and Yao, 
nonlocal game rigidity has been an extremely fruitful concept in quantum cryptography (see, e.g., Refs.~\cite{vazirani2019fully} and~\cite{CGJV19-verifiable-qc}), complexity theory (see Ref.~\cite{ji2020mip} and the references therein), and quantum information more generally~\cite{vsupic2020self}. This motivates the following question: what other tasks in quantum information also exhibit rigidity phenomena? 

To our knowledge, the only other work on rigidity phenomena outside of nonlocal games is that reported in Refs.~\cite{tavakoli2018self,farkas2019self} on the rigidity of \emph{quantum random access codes (QRACs)}. 
The authors study ``$2^d \rightarrow 1$'' QRACs, which encode~$2$ classical dits $x,y \in [d]$ into a $d$-dimensional system, such that either~$x$ or~$y$ may be retrieved by performing a suitable measurement. 
These works show that~$2^d \to 1$ QRACs are rigid, and in fact certify measurements based on mutually unbiased bases (MUBs). 

In this paper we investigate the rigidity properties of superdense coding, which plays a fundamental role in quantum Shannon theory (see, e.g., Ref.~\cite[Chapter~6]{wilde2013quantum}). The superdense coding \emph{task} is to communicate one of four possible messages while only transmitting one quantum bit across a channel. The superdense coding \emph{protocol}, first proposed by Bennett and Wiesner~\cite{bennett1992communication}, achieves this task in the following way: Alice and Bob share one qubit each of an EPR pair (i.e., the maximally entangled state $\tfrac{1}{\sqrt{2}} \ket{00} + \frac{1}{\sqrt{2}}  \ket{11} $) in advance, and to transmit a message $i \in \{1,2,3,4\}$, Alice applies a one of four Pauli operators to her half of the EPR pair and sends her qubit. Bob then performs a Bell measurement on the qubit received from Alice and his qubit to determine $i$. 

\subsection{Rigidity for superdense coding of two classical bits}

The first result in our paper is to show that superdense coding is rigid: \emph{any} protocol that accomplishes this task is ``locally equivalent'' to the Bennett-Wiesner protocol. We model arbitrary protocols for superdense coding in the following manner: Alice and Bob share a density matrix $\tau$ on a bipartite Hilbert space $\Hilb_A \otimes \Hilb_B$, where we assume without loss of generality that $\Hilb_{A}$ factors into $\Hilb_{A'} \otimes \Hilb_{A''}$ where $\Hilb_{A''}$ is isomorphic to $\complex^2$. Given an input $i \in \{1,2,3,4\}$, Alice applies a unitary operator $U_i$ (called an \emph{encoding operator\/}) to her share of $\tau$ (with support in the space~$\Hilb_A$), sends the qubit~$A''$ to Bob, and Bob then performs an optimal distinguishing measurement on the Hilbert space $\Hilb_{A''} \otimes \Hilb_{B}$ to determine what the input $i$ was. See \Cref{fig:general-protocol} for an illustration of a general superdense coding protocol.

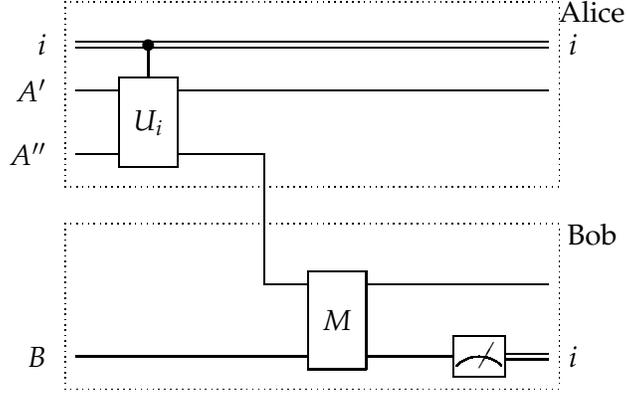
\begin{figure}
\centerline{
\Qcircuit @C=1.5em @R=1em {
                                 &                          &                          &                 &              &              &                  &                                    &   \mbox{Alice}   \\
   \lstick{i\rule{0.5em}{0em}}   &   \control \cw \qwx[1]   &   \cw                    &   \cw           &   \cw        &   \cw        &   \cw            &   \cw                              &   \lstick{i}     \\
   \lstick{A'\rule{0.5em}{0em}}   &   \multigate{1}{U_i}         &   \qw   			   &   \qw           &   \qw        &   \qw        &   \qw            &   \rstick{\rule{0.5em}{0em}} \qw   &       \\
   \lstick{A''\rule{0.5em}{0em}}   &   \ghost{U_i}               &   \qw                  &   \qw \qwx[1]   &              &              &                  &   \push{\rule{0em}{1.5em}}         &                  \\
                                 &                          &                          &   \qwx[1]       &              &              &                  &                                    &                  \\
                                 &                          &                          &   \qwx[1]       &              &              &                  &                                    &   \mbox{Bob}     \\
   \push{\rule{0em}{1.5em}}      &                          &                          &                 &   \multigate{1}{M}   &   \qw   &   \qw   &   \qw                              &      \\
   \lstick{B\rule{0.5em}{0em}}   &   \qw                    &   \qw                    &   \qw           &   \ghost{M}      &   \qw        &   \meter   &   \push{\rule{0em}{1.5em}}  \cw    &   \lstick{i}
\gategroup{1}{1}{4}{8}{.7em}{.}
\gategroup{6}{1}{8}{8}{.7em}{.}
}
}
\caption{A general superdense coding protocol. This quantum circuit is modified from~\cite{Charezma}.} \label{fig:general-protocol}
\end{figure}

\emph{A priori\/} it appears daunting to characterize the structure of an arbitrary superdense coding protocol. For one, the dimension of the spaces $\Hilb_{A'}$ and $\Hilb_B$ are unbounded, and the state $\tau$ is uncharacterized. Furthermore, the encoding unitary operators $U_i$ can be extremely complicated, potentially performing complex entangling operations between the space $\Hilb_{A'}$ and $\Hilb_{A''}$ (the qubit to be sent over to Bob). However, the property of being a superdense coding protocol is extremely constraining. \Cref{thm:rigidity} gives a precise characterization of how an arbitrary superdense protocol is \emph{locally equivalent} to the canonical Bennett-Wiesner protocol. In the statement of the theorem, ``$=_{\tau'}$'' denotes equality of two unitary operators with respect to the state $\tau'$; in other words, $C =_{\rho} D$ means that $C \rho C^* = D \rho D^*$.

\begin{restatable}[Rigidity for superdense coding]{theorem}{rigidity}
\label{thm:rigidity}

Let $(\tau,(U_i))$ denote a superdense coding protocol. Then there exist 
\begin{enumerate}
	\item Unitary operators $V$ acting on $\Hilb_{A'} \otimes \Hilb_{A''}$ and $(C_i)_{i \in [4]}$ acting on $\Hilb_{A'}$,
	\item An isometry $W$ mapping $\Hilb_B$ to a Hilbert space $\Hilb_{B'} \otimes \Hilb_{B''}$ where $\Hilb_{B''}$ is isomorphic to $\C^2$,
	\item A density matrix $\rho$ on $\Hilb_{A'} \otimes \Hilb_{B'}$,
	\item A set of pairwise orthogonal projectors $\{ P_r \}$ that sum to the identity on $\Hilb_{A'}$, and
	\item A collection of~$2 \times 2$ unitary operators $\{S_r\}$,
\end{enumerate}
such that, letting $\tau' \eqdef (V \otimes W) \tau (V \otimes W)^*$, we have
\[
	\tau' = \rho^{A'B'} \otimes \ketbra{\EPR}{\EPR}^{A'' B''}
\]
and for $i \in \{1,2,3,4\}$,
\[
	(C_i^* \otimes \id) U_i V^* =_{\tau'} \sum_{r} P_r \otimes S_r \sigma_i S_r^* 
\]
where $\sigma_1 \eqdef \id$, $\sigma_2 \eqdef \PauliZ$, $\sigma_3 \eqdef \PauliX$, and $\sigma_4 \eqdef \PauliY$ are the one-qubit Pauli matrices. 

\end{restatable}

\Cref{thm:rigidity} can be interpreted as expressing rigidity for superdense coding in the following way: given an arbitrary protocol $(\tau,(U_i))$ for superdense coding, there exists local isometries $V, W$ where if Alice applies $V$ and Bob applies $W$ to their share of $\tau$, then an EPR pair is extracted with an auxiliary state $\rho$ remaining. By pre-applying $V^*$ to Alice's unitary operators $U_i$, we discover that $U_i V^*$ has a very regular form: operationally, it can be interpreted as performing some projective measurement $\{ P_r\}$ on Alice's part of the auxiliary state $\rho$ to obtain some outcome $r$ in a set $\cal{R}$, and then based on $r$, applying a rotated version of the standard Bennett-Wiesner superdense coding protocol to Alice's part of the EPR pair. Finally, after sending her EPR qubit, Alice then applies some unitary operator $C_i$ on her remaining qubits (which does not affect Bob's measurement in any way). This considerably strengthens and extends the characterization of ``tight'' superdense coding protocols due to Vollbrecht and Werner~\cite[Lemma~3]{VW00-Pauli-operators} (see also Ref.~\cite{Werner01-teleportation-dense-coding}); they studied protocols in which the shared entangled state~$\tau$ is a state on~$\complex^2 \tensor \complex^2$ (or on~$\complex^d \tensor \complex^d$ in the case of~$d$-dimensional superdense coding protocols; see \Cref{sec-intro-d-dim}). This difference would be significant in a cryptographic setting; in the context of quantum key distribution, this is the difference between the \emph{device-independent\/} and \emph{semi-device-independent\/} settings.

The proof of \Cref{thm:rigidity} is given in \Cref{sec:d2}. It proceeds via a number of reductions: first, using an information-theoretic argument, we show every superdense coding protocol $(\tau,(U_i))$ is locally equivalent to one that uses an EPR pair in the state $\tau$. Given this, we then show that each of the encoding operators $U_i$ can be individually block-diagonalized with respect to the EPR pair. Finally, we show that the blocks across the different encoding operators $U_i$ can be ``matched up'' in a way that they correspond to the Pauli matrices. Each of these steps requires carefully deducing the structure imposed on the state and the encoding operators by the correctness of the protocol.

\subsection{Rigidity for higher dimensional superdense coding?}
\label{sec-intro-d-dim}

We then consider the generalization of superdense coding to communicating more than $2$ classical bits. Specifically, we consider protocols for communicating one of $d^2$ possible messages by sending a $d$-dimensional quantum system over the channel --- we call these $d$-dimensional superdense coding protocols. A canonical protocol for $d$-dimensional superdense coding is as follows: the players share a $d$-dimensional maximally entangled state $\ket{\mes_d} \eqdef \frac{1}{\sqrt{d}} \sum_{e=0}^{d-1} \ket{e} \ket{e}$, and given message $i \in [d^2]$, Alice applies a unitary operator $E_i$ to her share of $\ket{\mes_d}$, and sends it over to Bob. The family of unitary operators $\{E_i\}$ can be any orthogonal unitary basis for the space of $d \times d$ matrices. (The orthogonality property means that $\Tr(E_i^\adjoint E_j) = 0$ if and only if $i \neq j$.) An example of such a basis is the set of Heisenberg-Weyl operators. In dimension $d$, these are a set of $d \times d$ matrices $\{ P_{i,j} : 0 \le i,j < d \}$ defined as follows. Let $\omega_d \eqdef \exp \Paren{ \frac{2 \pi \complexi}{d}}$ be a primitive $d$th root of unity. For $i,j \in \{0,1,2,\ldots,d-1\}$, let $P_{i,j} = \PauliX^i_d \PauliZ^j_d$ where $\PauliX_d \eqdef \sum_{k=0}^{d-1} \ketbra{k+1 \, (\mathrm{mod} \, d)}{k}$ is the ``shift'' operator, and $\PauliZ_d \eqdef \sum_{k=0}^{d-1} \omega^k_d \, \ketbra{k}{k}$ is the ``clock'' operator.  Does the rigidity phenomenon also extend to dimensions~$d$ larger than~$2$? 

The second result of this paper is that $d$-dimensional superdense coding for $d \geq 3$ is not rigid in the same sense as \Cref{thm:rigidity}: there are $d$-dimensional superdense coding protocols which are not locally equivalent to each other. This is because in dimensions three and higher there are \emph{inequivalent\/} orthogonal unitary bases. (In contrast, all orthogonal unitary bases in dimension two are equivalent to the Pauli matrices.) Here, equivalence between two unitary bases $\{E_i\}$ and $\{F_i\}$ means there exist unitary operators $U,V$ such that for all $i$, we have $F_i = \alpha_i U E_i V$ for some choice of complex phase $\alpha_i$.

\begin{restatable}[Existence of inequivalent orthogonal unitary bases for all $d \geq 3$]{theorem}{ineq-oub}
\label{thm:ineq-oub-intro}
For every dimension~$d \ge 3$, there are orthogonal unitary bases that are not equivalent to each other.
\end{restatable}

The uniqueness of orthogonal unitary bases was first studied by Vollbrecht and Werner~\cite{VW00-Pauli-operators}, and the existence of non-equivalent orthogonal unitary bases for all dimensions greater than~$2$ was observed in follow-up work by Werner~\cite{Werner01-teleportation-dense-coding}. (We elaborate on prior work on the topic in \Cref{sec-oub-uniqueness}.) Werner described, without proof, how non-equivalent bases may be constructed. We present explicit constructions of such bases in \Cref{sec-oub}. The construction for~$d \ge 4$ is based on the observation that the shift operator~$\PauliX_d$ corresponds to a perfect matching in~$K_{d,d}$, the complete bipartite graph. Moreover, its powers~$\set{ \PauliX_d^i : 0 \le i < d }$ correspond to a partition of the edge set of~$K_{d,d}$ into~$d$ \emph{disjoint\/} perfect matchings. By replacing this partition with another carefully chosen such partition, we obtain an orthogonal unitary basis that is not equivalent to the clock and shift construction. The proof of non-equivalence involves comparing the spectra of the operators in the two bases, taking into account the complex phase and unitary operators that witness a potential equivalence map. For~$d = 3$, we follow a construction described by Werner~\cite{Werner01-teleportation-dense-coding}. We prove non-equivalence to the clock and shift basis by showing that the resulting basis is not a commutative projective group (again accounting for a potential equivalence map).

In a previous version of this paper, we conjectured that rigidity for higher dimensional superdense coding holds \emph{up to\/} choosing orthogonal unitary bases~\cite[Conjecture~1.3]{NY20-rigidity}. That is, every $d$-dimensional superdense coding protocol is locally equivalent (in the sense of \Cref{thm:rigidity}) to one where Alice and Bob share an entangled state $\rho^{A'B'} \otimes \ketbra{\mes_d}{\mes_d}^{A'' B''}$ for some density matrix $\rho^{A'B'}$, and Alice's encoding operators are of the form 
\[
	U_i = \sum_r P_r \otimes E_{r,i}
\]
where $\{P_r\}$ is a set of pairwise orthogonal projectors that sum to the identity on $\Hilb_{A'}$ and for every~$r$, the set $\{E_{r,i}\}_{i \in [d^2]}$ is an orthogonal unitary basis for the space of $d \times d$ complex matrices. This would be a natural extension of the statement of \Cref{thm:rigidity} to the case of general $d \geq 2$ where the registers $A' B'$ are treated as a source of ``shared randomness'' to help Alice and Bob synchronize their choice of orthogonal unitary basis. 

We can show that when the shared entangled state between Alice and Bob is a pure state in $\C^d \otimes \C^d$, then this conjecture holds (see \Cref{sec-oub-connection}): up to local unitary operators, the shared state is necessarily the maximally entangled state $\ket{\mes_d}$ of local dimension $d$, and the encoding operators $\{U_i\}$ necessarily form an orthogonal unitary basis. However, this conjecture is false for protocols where Alice sends only a \emph{part} of her entangled state. In work subsequent to ours, Farkas, Kaniewski, and Nayak~\cite{FKN22-mum-superdense-coding} show that there exist infinitely many superdense coding protocols that are not locally equivalent to a protocol of the form described in the conjecture. In particular, in these counterexample protocols Alice may perform a complicated entangling operation between her message and the rest of her state, rather than just treating the ancilla system as a source of shared randomness.

\subsection{Superdense coding protocols with error}
\label{sec-robust-rigidity}

Finally, we consider \emph{probabilistic} protocols for $d$-dimensional superdense coding, where Bob's decoding only needs to succeed with high probability. In particular, we say that $(\tau,(U_i))$ is a $(d,\eps)$-superdense coding protocol if Bob is able to decode Alice's message $i$ with probability at least $1 - \eps$, for all $i$. We focus on the case where Alice and Bob share an entangled state in $\C^d \otimes \C^d$ (i.e., have local dimension $d$). As mentioned previously, in the exact case $\eps = 0$, their shared state is necessarily the maximally entangled state and Alice's encoding unitary operators form an orthogonal unitary basis. We conjecture that even in the probabilistic setting, this characterization of $d$-dimensional superdense coding protocols is \emph{robust\/}, in the following sense.
\begin{conjecture}
\label{conj:robust-rigidity}
There exist functions $\delta_1,\delta_2: [0,1] \to [0,1]$ where $\delta_1(\eps)$ and $\delta_2(\eps)$ monotonically decrease to $0$ as $\eps \to 0$, such that the following holds. For all $(d,\eps)$-superdense coding protocols $(\tau,(U_i))$ such that $\tau$ is a density matrix on $\C^d \otimes C^d$ and $U_i$ are unitary operators in $\unitary(\C^d)$, we have 
\[
	\bra{\mes_d} \tau \ket{\mes_d} \quad \geq \quad 1 -  \delta_1(\eps) \enspace,
\]
and there exists an orthogonal unitary basis $\{E_i\}_{i \in [d^2]}$ for the space of $d \times d$ complex matrices such that for all $i \in [d^2]$,
\[
	\| U_i - E_i \|_{\nhs} \quad \leq \quad \delta_2(\eps) \enspace,
\]
where $\| X\|_{\nhs} \eqdef  \sqrt{\frac{1}{d} \Tr(X X^\dagger)}$ denotes the normalized Hilbert-Schmidt norm on the space of $d \times d$ matrices.
\end{conjecture}
We note that the choice of the normalized Hilbert-Schmidt norm in the statement of \Cref{conj:robust-rigidity} is somewhat arbitrary; one can also consider other formulations of the conjecture with other norms (such as the spectral norm, etc.).  

The last part of our paper analyzes a possible challenge to \Cref{conj:robust-rigidity} proposed by Aram Harrow. %
Consider the following probabilistic construction for a potential $(d,\eps)$-superdense coding protocol:
independently sample $d^2$ matrices $\bU_1,\ldots, \bU_{d^2}$ from the Haar measure on~$\unitary( \complex^d)$, the group of $d \times d$ complex unitary matrices. Let $\tau \eqdef \ketbra{\mes_d}{\mes_d}$ denote the $d$-dimensional maximally entangled state. How well does the protocol $(\tau,(\bU_i))$ accomplish superdense coding? 

In classical and quantum communication, many tasks can be performed near-optimally via probabilistically constructed protocols. See, e.g., the text by Wilde~\cite{wilde2013quantum} for examples from Shannon theory. A simple example from communication complexity is the task of \emph{quantum fingerprinting\/}~\cite{buhrman2001quantum}, which enables checking whether two $n$-bit strings $x$ and $y$ are equal by only comparing two $\Order(\log n)$-qubit \emph{fingerprints} of the strings. It can be shown that picking random $\Order(\log n)$-qubit states for each $n$-bit string $x$ yields a good quantum fingerprinting protocol. 

Let $\bPi_d$ denote the random superdense protocol specified by~$( \mes_d, (\bU_i))$. Note that the error $\beps$ of $\bPi_d$, when averaged over the choice of random unitaries $(\bU_i)$, is some function of $d$. We first argue that the conjecture implies that the error of a random superdense coding protocol, when averaged over the choice of $(\bU_i)$, cannot be too small:

\begin{restatable}{proposition}{counterexample}
\label{prop:counterexample}
Suppose \Cref{conj:robust-rigidity} were true. Let $\delta_2(\eps)$ be the function from \Cref{conj:robust-rigidity}. Then the random superdense coding protocol~$\bPi_d$ specified by~$( \mes_d, (\bU_i))$ must have error $\beps$ satisfying
\[
	\Ex_{(\bU_i)} \delta_2(\beps)^2 \quad \geq \quad (2d)^{-2} \enspace.
\]
\end{restatable}
Put another way, it cannot be that both \Cref{conj:robust-rigidity} is true and also the random superdense protocol has error vanishing so quickly such that $\delta_2(\beps)$ is smaller than $(2d)^{-2}$, on average. Due to the concentration of measure phenomenon for Haar-random states and unitary operators (as expressed by, e.g., the L\'{e}vy-like property in \Cref{thm-Haar-concentration}), it is plausible \emph{a priori\/} that the average error $\beps$, and therefore also~$\delta_2(\beps)$, scale as $\order(d^{-2})$. Thus, the random superdense protocol is potentially a counterexample to \Cref{conj:robust-rigidity}.

We show that this probabilistic construction does \emph{not\/} yield a good superdense coding protocol: with overwhelmingly high probability over the choice of random unitary operators $(U_i)$, the protocol has a nonzero probability of error that is independent of $d$. Thus, \Cref{conj:robust-rigidity} is not ruled out by the random protocol construction.

\begin{restatable}[Performance of a random superdense coding protocol]{theorem}{random-protocol}
\label{thm:random-protocol}
The random superdense coding protocol~$\bPi_d$ specified by~$( \mes_d, (\bU_i))$ where~$\bU_i \in \unitary( \complex^d)$ are Haar-random unitary operators has error at least~$1 - \tfrac{8}{3\pi} \approx 0.15$ as~$d \to \infty$, with high probability over the choice of $(\bU_i)$. 
\end{restatable}

We prove \Cref{thm:random-protocol} by showing that the \emph{distinguishability\/} of the ensemble of random states $\{ (\bU_i \otimes \id) \ket{\mes_d} \}_{i \in [d^2]}$ is bounded away from $1$ (with high probability). The generalized Holevo-Curlander bounds~\cite{Holevo79-distinguishability,Curlander79-distinguishability,ON99-converse-channel-coding,Tyson09-distinguishability} relate the distinguishability of an ensemble $\{ (p_i, \rho_i) \}$ to the quantity
\begin{equation}
\label{eq-hc-bound}
\Tr \Big( \sum_i p_i^2 \rho_i^2 \,\Big)^{1/2} \enspace.
\end{equation}
Our analysis of this quantity is largely inspired by work due to Montanaro~\cite{montanaro2007distinguishability} on the distinguishability of random pure quantum states. However, extending his approach to the ensemble of interest to us---one consisting of random \emph{maximally entangled\/} states---involves significant technical difficulties. The approach involves relating the distinguishability of an ensemble of states to the spectrum of the ensemble average. In the case of  Haar-random pure states, the ensemble average is well approximated by the ensemble average of \emph{unnormalized\/} complex gaussian vectors with suitably chosen variance. The spectrum of such matrices in the asymptotic limit is given by the Mar\v{c}enko-Pastur Theorem from random matrix theory. In our case, the entries of the random vectors in the ensemble are \emph{not\/} independent. We instead bound the generalized Holevo-Curlander quantity in \Cref{eq-hc-bound} by employing a recent generalization of the Mar\v{c}enko-Pastur Theorem due to Yaskov~\cite{Yaskov16-MP-short-proof}. (The theorem was proven for ensembles of random real vectors. We verify that its proof also extends to complex random vectors with analogous properties.) In the process, we show that random maximally entangled states satisfy a \emph{pseudo-isotropy\/} condition that suffices for the theorem to hold. 

A subtlety in the use of the Mar\v{c}enko-Pastur law is that we would like to deduce the convergence of a sequence of means to the mean of the limiting distribution from the convergence of a sequence of distributions. This does not necessarily hold in general. In order to prove such a relation between the two forms of convergence, we show that random maximally entangled states are \emph{sub-gaussian\/}. This allows us to draw on a generalization of the Bai-Yin Theorem, which bounds the norm of matrices whose columns are given by i.i.d.\ sub-gaussian vectors. We thus show that the norm of the ensemble average has an exponentially decaying tail, which in turn guarantees the form of convergence we seek.

We believe the techniques used in our analysis are of independent interest. In fact, the subtlety mentioned above was overlooked by Montanaro; the ideas we develop may also be used to close a gap in his analysis (see Section~\ref{sec-expectation-limit} for the details).

\subsection{Further remarks and open questions}

In this paper we have initiated the study of rigidity phenomena in superdense coding protocols. Given its importance in quantum Shannon theory, our study may shed new light on protocols based on superdense coding. The power of entanglement as a resource in distributed quantum computation, in particular in two-party communication complexity, remains a mystery. The rigidity theorem we establish (Theorem~\ref{thm:rigidity}) gives a complete picture for a simple but fundamental task. The property shown in the analysis of random superdense coding protocols (Theorem~\ref{thm:random-protocol}) may also be interpreted as placing a limit on how closely a sequence of random unitary operators approximate an orthogonal basis. This may be of relevance to the theory of error-correction, where unitary error bases play a central role. 

We list several open questions that arise from this work:
\begin{enumerate}
	\item Is \Cref{conj:robust-rigidity} true? Does a robust version of \Cref{thm:rigidity} hold? 
	\item Do $d$-dimensional superdense coding protocols, in which the shared state between Alice and Bob may have local dimension larger than~$d$, also exhibit some non-trivial form of rigidity?
	\item Does rigidity also hold for quantum teleportation, a task that is ``dual'' to superdense coding? Can this be derived in a black-box way from the rigidity of superdense coding?
	\item Are there any connections between the QRAC rigidity results of~\cite{tavakoli2018self,farkas2019self} and our results on the rigidity of superdense coding?
	\item What other quantum information processing tasks have the rigidity property?
\end{enumerate}
We believe the investigation of these questions will lead to significant new insights into the nature of quantum information, with wide ranging ramifications.

\paragraph{Acknowledgements.}
We thank the anonymous journal reviewers for their thorough feedback on an earlier version of the paper. We thank Adam Bouland, Chinmay Nirkhe, and Zeph Landau for stimulating discussions at the beginning of this project. H.Y.\ would like to especially thank Adam Bouland for his integral role in formulating the questions explored in this paper. We would like to thank Pavel Yaskov for the correspondence on his work on the Mar\v{c}enko–Pastur theorem.
We thank J{\k{e}}drzej Kaniewski and Mate Farkas for their pointers to the literature on rigidity of QRACs. A.N.\ would like to thank Kanstantsin Pashkovich and Vern Paulsen for helpful discussions on bases of orthogonal unitary operators, and is grateful to the Berkeley CS Theory Group for their hospitality during a visit in Fall~2017, when this work was initiated. A.N.'s research is supported in part by a Discovery Grant from NSERC Canada.
H.Y.\ is supported by an NSERC Discovery Grant, a Google Quantum Research Award, and AFOSR Grant No.\! FA9550-21-1-0040. 
This research was partly conducted at the Kavli Institute for Theoretical Physics during the Quantum Physics of Information program in 2017 (and thus this research was supported in part by the National Science Foundation under Grant No. NSF PHY17-48958).

\section{Properties of superdense coding}
\label{sec:prelim}

\subsection{Quantum information basics} 

We refer the reader to texts such as~\cite{NC11-quantum-information,Watrous18-TQI,wilde2013quantum} for the basics of quantum information, and mention some notational conventions here.

We write $\id$ to denote the identity operator on a Hilbert space. We use superscripts on quantum states, e.g., $\ket{\psi}^{AB}$ or $\rho^{AB}$, to denote the registers in which they are stored. Similarly, we use subscripts on operators to indicate the registers on which they act, unless this is clear from the context. Given a bipartite density matrix $\rho^{AB}$, we write $\rho^A$ to denote its reduction to register $A$ (i.e., the partial trace over~$B$).

Given operators $A, B$ and a density matrix $\rho$ acting on a Hilbert space $\Hilb$, we write $A =_\rho B$ to denote $A\rho A^* = B\rho B^*$. In other words, the operators $A$ and $B$ have the same action on the state $\rho$.

We write $\ket{\EPR}$ to denote the maximally entangled state $\frac{1}{\sqrt{2}} \Big( \ket{00} + \ket{11} \Big)$ on two qubits.  We write $\ket{\mes_d} = \frac{1}{\sqrt{d}} \sum_{i=0}^{d-1} \ket{i} \ket{i}$ to denote the $d$-dimensional maximally entangled state, or simply $\ket{\mes}$ if the dimension $d$ is clear from context. We recall the single qubit Pauli matrices:
\[
	\id \eqdef \begin{pmatrix} 1 & 0 \\ 0 & 1 \end{pmatrix} \qquad \PauliX \eqdef \begin{pmatrix} 0 & 1 \\ 1 & 0 \end{pmatrix} \qquad  \PauliY \eqdef \begin{pmatrix} 0 & - \complexi \\ \complexi & 0 \end{pmatrix} \qquad  \PauliZ \eqdef \begin{pmatrix} 1 & 0 \\ 0 & -1 \end{pmatrix} \;.
\]

\subsection{Basic properties of superdense coding}

Here we give a formal definition of a general superdense coding protocol, and prove some basic properties about them.

\begin{definition}[Superdense coding protocol]
\label{def:protocol}
Let $d$ be a positive integer.
\suppress{
Let $\Hilb_A \eqdef \Hilb_{A'} \otimes \Hilb_{A''},\Hilb_B \eqdef \Hilb_{B'} \otimes \Hilb_{B''}$ be finite dimensional Hilbert spaces where $\Hilb_{A'}$ is isomorphic to $\Hilb_{B'}$ and $\Hilb_{A''},\Hilb_{B''}$ are both isomorphic to $\complex^d$. 
}
Let $\Hilb_A \eqdef \Hilb_{A'} \otimes \Hilb_{A''}$ and~$\Hilb_B$ be finite dimensional Hilbert spaces where $\Hilb_{A''}$ is isomorphic to $\complex^d$. 
Let $\tau$ denote a density matrix on $\Hilb_A \otimes \Hilb_B$ and let $(U_i)_{i \in [d^2]}$ denote a sequence of $d^2$ unitary operators (called \emph{encoding operators}) acting on $\Hilb_A$. We say that $(\tau,(U_i))$ is a \emph{$(d,\eps)$-superdense coding protocol} if there exists a POVM $\{ M_i \}_{i \in [d^2]}$ acting on $\Hilb_{A''} \otimes \Hilb_B$ such that
\begin{equation}
\label{eq:protocol-povm}
	\Tr(M_i \, \rho_i) \geq 1 - \eps \qquad \forall \, i \in [d^2]
\end{equation}
where $\rho_i$ denotes the reduced density matrix of $(U_i \otimes \id) \tau (U_i \otimes \id)^*$ on registers $A'' B$. When $\eps = 0$ we simply call $(\tau,(U_i))$ a \emph{$d$-dimensional superdense coding protocol}.
\end{definition}

\begin{lemma}[Orthogonality conditions I]
\label{lem:orthog0}
Let $(\tau,(U_i))$ be a $d$-dimensional superdense coding protocol. Then letting $\rho_i$ denote the reduced density matrix of $(U_i \otimes \id) \tau (U_i \otimes \id)^*$ on registers $A'' B$, we have that
\[
	\Tr(\rho_i \, \rho_j) = 0 \qquad \forall i \neq j \in [d^2]\;.
\]
\end{lemma}
\begin{proof}
	Let $\{M_i\}$ denote a POVM satisfying \Cref{eq:protocol-povm} for $\eps = 0$. Then for all $i \in [d^2]$, we have $\rho_i \leq M_i$ according to the positive semidefnite ordering. This is because if we write $\rho_i = \sum_k p_{ik} \ketbra{\psi_{ik}}{\psi_{ik}}$ for some probabilities $\{p_{ik}\}$, then $\Tr(\rho_i M_i) = 1$ implies that for all $k$, $\bra{\psi_{ik}} M_i \ket{\psi_{ik}} = 1$, which implies that $\ket{\psi_{ik}}$ is an eigenvector of $M_i$ with eigenvalue $1$. This implies that $M_i = \sum_k \ketbra{\psi_{ik}}{\psi_{ik}} + M_i'$ for some positive semidefinite operator $M_i'$, and this is at least $\rho_i$ in the positive semidefinite ordering. 
	
	This means then that for all $j \neq i$, 
	\[
		0 \leq \Tr(\rho_i \, M_j) \leq \Tr(\rho_i \, (\Id - M_i)) = \Tr(\rho_i) - \Tr(\rho_i \, M_i) = 0
	\]
	where the first inequality is due to the positivity of $\rho_i$ and $M_j$, the second inequality is due to the fact that $\sum_i M_i = \id$, and the last equality is due to the fact that $\Tr(\rho_i) = \Tr(\rho_i \, M_i) = 1$. Therefore $\Tr(\rho_i \, M_j) = 0$. Thus we have
	\[
		0 \leq \Tr(\rho_i \, \rho_j) \leq \Tr(\rho_i \, M_j) = 0 \enspace,
	\]
	so $\Tr(\rho_i \, \rho_j) = 0$. 
\end{proof}

\begin{lemma}[Orthogonality conditions II]
\label{lem:orthog}
Let $(\tau,(U_i))$ be a $d$-dimensional superdense coding protocol. Then for all $i \neq j \in [d^2]$, 
	\[
		\Tr_{A''} \left ( U_i \tau^A U_{j}^* \right) = 0
	\]
	where $\tau^A$ denotes the reduced density matrix of $\tau$ on register $A$ and $\Tr_{A''}(\cdot)$ denotes the partial trace over register $A''$. 
\end{lemma}
\begin{proof}

Let $\ket{\tau}$ denote a purification of $\tau$ on the Hilbert space $\Hilb_A \otimes \Hilb_B \otimes \Hilb_R$, where $\Hilb_R$ is the purifying space. Clearly, the protocol where Bob also has access to the purifying space $\Hilb_R$ is also a $d$-dimensional superdense coding protocol. 

Let $\{ \ket{1}, \ldots, \ket{\dim A'} \}$ denote an orthonormal basis for $\Hilb_{A'}$. Let~$\ket{\rho_{ik}}$ be the (sub-normalized) pure state on registers $A'' B R$ given by
\[
	\ket{\rho_{ik}} \eqdef (\bra{k}^{A'} \otimes \id)(U_i \otimes \id) \ket{\tau} \enspace.
\]
Intuitively, $\ket{\rho_{ik}}$ represents the residual state of the protocol on registers $A'' B$ when Alice applies unitary operator $U_i$, and then measures the $A'$ subsystem in the standard basis to obtain outcome $\ket{k}$. 

Note that if we let $\rho_i$ denote the state $(U_i \otimes \id) \ketbra{\tau}{\tau} (U_i \otimes \id)^*$ reduced to the registers $A'' B R$, we have the identity
\[
	\rho_i = \sum_k \ketbra{\rho_{ik}}{\rho_{ik}}\;,
\]
because we can think of $\rho_i$ as the result of first measuring the $A'$ register in the standard basis, and discarding the outcome. Applying \Cref{lem:orthog0} to the purified protocol $(\ketbra{\tau}{\tau}, (U_i))$, we have for $i \neq j$,
\[
	0 = \Tr(\rho_i \, \rho_j) = \sum_{k,k'} \Tr( \ketbra{\rho_{ik}}{\rho_{ik}} \cdot \ketbra{\rho_{jk'}}{\rho_{jk'}} ) = \sum_{k,k'} |\braket{\rho_{jk'}}{\rho_{ik}}|^2
\]
and therefore $|\braket{\rho_{jk'}}{\rho_{ik}}|^2 = 0$ for all $k,k'$. This implies that $\braket{\rho_{jk'}}{\rho_{ik}} = 0$ for all $k,k'$, which can be rewritten as
\[
	\bra{\tau} (U_j \otimes \id)^* (\ketbra{k'}{k}^{A'} \otimes \id) (U_i \otimes \id)\ket{\tau} = \Tr \Big ( (\bra{k}^{A'} \otimes \id) (U_i \otimes \id) \ketbra{\tau}{\tau} (U_j^* \otimes \id) (\ket{k'}^{A'} \otimes \id) \Big) = 0.
\]
This is equivalent to the statement that
\[
\bra{k}^{A'} \, \Tr_{A''} \Big ( U_i \tau^A U_j^* \Big) \, \ket{k'}^{A'} = 0.
\]
Since this holds for all $k,k'$, the matrix $\Tr_{A''} \Big ( U_i \tau^A U_j^* \Big)$ is identically zero, which completes the proof of the Lemma.
\end{proof}

Next we define what it means for superdense protocols to be locally equivalent. 

\begin{definition}
\label{def:equiv-protocols}
Let $\Hilb_{A} \eqdef \Hilb_{A'} \otimes \Hilb_{A''}$ be a Hilbert space where $\Hilb_{A''}$ is isomorphic to $\C^d$. Let $\tau,\tau'$ be density matrices on $\Hilb_A \otimes \Hilb_B$. Let $(U_i), (U_i')$ be unitary operators acting on $\Hilb_A$. We say that $(\tau,(U_i))$ and $(\tau',(U_i'))$ are \emph{locally equivalent} if there exists
\begin{enumerate}
		\item A unitary operator $V$ acting on $\Hilb_{A'} \otimes \Hilb_{A''}$\;,
		\item A set of unitary operators $(C_i)_{i \in [d^2]}$ acting on $\Hilb_{A'}$
\end{enumerate}
such that
\begin{enumerate}
	\item $\tau' = (V \otimes \id) \tau (V \otimes \id)^*$, and
	\item $U_i' = (C_i \otimes \id) U_i V^*$.
\end{enumerate}
\end{definition}

\begin{lemma}[Local unitary freedom of superdense coding protocols]
\label{lem:gauge_freedom}
The following properties hold for local equivalence.
	\begin{enumerate}
		\item If $(\tau,(U_i))$ and $(\tau',(U_i'))$ are locally equivalent, then $(\tau,(U_i))$ is a $(d,\eps)$-superdense coding protocol if and only if $(\tau',(U_i'))$ is.
		\item Local equivalence is transitive.
	\end{enumerate}
\end{lemma}
\begin{proof}

For any fixed $i$, after Alice applies her encoding unitary operator, the reduced density matrix on registers $A'' B$ is the same whether the protocol $(\tau,(U_i))$ or $(\tau',(U_i'))$ is used. Thus Bob's ability to distinguish between the different messages is exactly the same. This establishes Item 1.

If $(\tau,(U_i))$ and $(\tau',(U_i'))$ are locally equivalent, then $\tau' = (V \tensor \id) \tau (V^* \tensor \id)$, and $U_i' = (C_i \otimes \id) U_i V^*$. If $(\tau',(U_i'))$ and $(\tau'',(U_i''))$ are locally equivalent, then $\tau'' = (V' \tensor \id) \tau' ((V')^* \tensor \id)$, and $U_i'' = (C_i' \otimes \id) U_i' (V')^*$. Thus we have
\begin{align*}
\tau'' \quad & = \quad (V' V \tensor \id) \tau (V^* (V')^* \tensor \id) \enspace, \qquad \text{and} \\
U_i'' \quad & = \quad (C_i' C_i \otimes \id) U_i V^* (V')^* \enspace, 
\end{align*}
which implies that $(\tau,(U_i))$ is locally equivalent to $(\tau'',(U_i''))$. This establishes Item 2.
\end{proof}

\subsection{Nice form protocols}
\label{sec-nice-form}

In this section, we define \emph{nice form\/} protocols and then show that every superdense coding protocol is locally equivalent to one that has a nice form.

\begin{definition}
\label{def:nice-form}
A~$d$-dimensional superdense coding protocol $(\tau,(U_i))$ has a \emph{nice form} if 
	\begin{enumerate}
		\item \label{p1} $U_1 = \id$,
		\item \label{p2} There exists an isometry $W: \Hilb_B \to \Hilb_{B'} \otimes \Hilb_{B''}$ where $\Hilb_{B''}$ is isomorphic to $\C^d$ such that 
		\[
			(\id \otimes W) \tau (\id \otimes W)^* = \rho^{A' B'} \otimes \ketbra{\mes_d}{\mes_d}^{A'' B''}	
		\]		
		for some density matrix $\rho$ on $\Hilb_{A'} \otimes \Hilb_{B'}$.
		\item \label{p3} For all $i \in [d^2]$, we have that 
		\[
			U_i \Tr_{B}(\tau) U_i^* = U_i \left( \rho^{A'} \otimes \frac{\id}{d} \right) U_i^* = \rho^{A'} \otimes \frac{\id}{d} \;
		\]
		where $\rho^{A'}$ denotes the reduced density matrix of $\tau$ on $\Hilb_{A'}$.
		\item \label{p4} Let the spectral decomposition of $\rho^{A'}$ be $\sum_k \lambda_k \Pi_k$ where $\lambda_k > 0$ for all $k$, with~$\lambda_k$ distinct. Then for all $k$ and $i \neq j$, we have
		\[
			\Tr_{A''}((\Pi_k \otimes \id) \, U_i U_{j}^* \, (\Pi_k \otimes \id)) = 0.
		\]
	\end{enumerate}
\end{definition}
Item~\ref{p2} says that up to a local unitary operation, the two parties share a maximally entangled state (in addition to other entanglement), and Item~\ref{p3} turns out to be a consequence of this. Item~\ref{p4} is equivalent to saying that the encoding of distinct messages~$i \neq j$ are orthogonal to each other. The proof of \Cref{lem:nice_form} below may give the reader further intuition into these properties. 

In the proof of \Cref{lem:nice_form}, we make use of an information-theoretic argument that involves quantities such as von Neumann entropy $\entropy(A)$, conditional entropy $\entropy(A | B)$, and mutual information $\mi(A : B)$. For a comprehensive reference on these quantities and their basic properties, we recommend Wilde's textbook~\cite{wilde2013quantum}. 
It is an interesting question whether \Cref{lem:nice_form} can be proved \emph{without} making use of these information-theoretic quantities.

\begin{lemma}
\label{lem:nice_form}
	All superdense coding protocols $(\tau,(U_i))$ are locally equivalent to a superdense coding protocol $(\tau',(U_i'))$ that has a nice form. 
\end{lemma}
\begin{proof}
	We define a unitary operator $V$ acting on $\Hilb_{A'}$ and unitary operators $(C_i : i \in [d^2])$ acting on $\Hilb_{A'}$ such that, letting $\tau' \eqdef (V \otimes \id) \tau (V \otimes \id)^*$ and $U_i' \eqdef (C_i \otimes \id) U_i V^*$, the pair $(\tau',(U_i'))$ is a superdense coding protocol and has a nice form.
	
	Let $V \eqdef U_1$ and let $C_1 \eqdef \Id$. This already yields Item 1 of \Cref{def:nice-form}. %

	Let $\ket{\tau}^{AB R}$ be a purification of $\tau$ where $\Hilb_R$ is a reference system of dimension $\dim(\Hilb_{A} \otimes \Hilb_{B})$. Consider the cq-state
	\[
		\xi \eqdef \frac{1}{d^2} \sum_i \ketbra{i}{i}^X \otimes (U_i \tensor \id) \ketbra{\tau}{\tau}^{ABR} (U_i^* \tensor \id) \enspace,
	\]
	where the Hilbert space of register~$X$ is $\Hilb_X$. By \Cref{lem:gauge_freedom}, the protocol $(V \tau V^*,(U_i V^*))$ is a superdense coding protocol. Therefore the information contained in registers~$A'' B$ about~$X$ in the state~$\xi$ is
	\[
		\mi(X : A'' B)_\xi = 2 \log_2 d \enspace.
	\]
	Intuitively, this is because Bob can perfectly recover the value of $i \in [d^2]$, i.e., $2 \log_2 d$ bits of information, from the registers $A'' B$ of $\xi$. On the other hand, we have that
	\[
		\mi(X : B)_\xi = 0
	\]
	 because without the qubit register $A''$, Bob has no information about $X$ (the state of the register~$B$ is the same for all~$i$). Therefore we get
	\[
		\mi(X: A'' |B )_\xi = \mi(X : A'' B)_\xi - \mi(X: B)_\xi = 2 \log_2 d \enspace.
	\]
	Using the entropy characterization of conditional mutual information, we get
	\[
		2 \log_2 d = \mi(X: A'' |B )_\xi = \entropy(A'' | B)_\xi - \entropy(A'' | XB)_\xi \enspace.
	\]
	Since $\entropy(A''|B)_\xi \leq \log_2 d $ and $\entropy(A'' | XB)_\xi \geq -\log_2 d $ (because the dimension of register $A''$ is $d$), we get that $\entropy(A''|B)_\xi = \log_2 d $ and $\entropy(A''| XB)_\xi = -\log_2 d $. 
	
	Since $X$ is a classical register, we can write $\entropy(A''|XB)_\xi$ as
	\[
		-\log_2 d  = \entropy(A''|XB)_\xi = \Ex_i \entropy(A'' |B, X = i)
	\]
	where $\entropy(A'' | B, X = i)$ is defined as $\entropy(A''|B)_{\xi_i}$ with $\ket{\xi_i} \eqdef (U_i \tensor \id) \ket{\tau}$. Since $\entropy(A''| B,X =i) \geq -\log_2 d $, we have $\entropy(A''|B)_{\xi_i} = -\log_2 d $ for all $i$, and in particular for $i = 1$.
	
	Then $\entropy(A''|B)_{\xi_i} = -\entropy(A'' | RA')_{\xi_i}$ (because $\xi_i$ is pure), so $\entropy(A'' | RA')_{\xi_i} = \log_2 d $. On one hand, we have that $\mi(A'' : RA')_{\xi_i} = \entropy(A'')_{\xi_i} - \entropy(A'' | RA')_{\xi_i}$, and on the other hand, mutual information is always nonnegative. Thus $\entropy(A'')_{\xi_i} = \log_2 d $, and the reduced density matrix of $\xi_i$ on the $A''$ register is maximally mixed. Furthermore we have $\mi(A'' : RA')_{\xi_i} = 0$, so $\xi_i$ has no correlations between registers $A''$ and $RA'$:
	\begin{equation}
	\label{eq:nice_form1}
		\Tr_B(\xi_i) = \rho_i^{RA'} \otimes \frac{\id}{d}
	\end{equation}
	where $\rho_i$ is some density matrix on the $RA'$ registers. 
	
	Fix $i = 1$, and let $\rho^{RA'}$ denote $\rho_1^{RA'}$. %
	Let $\Hilb_{B'}$ be a Hilbert space with dimension $\dim(\Hilb_{R} \otimes \Hilb_{A'})$ and let $\Hilb_{B''}$ be isomorphic to $\C^d$. Let $\ket{\rho}^{RA'B'}$ denote a purification of $\rho^{RA'}$. Notice that $\ket{\rho}^{RA'B'} \otimes \ket{\mes_d}^{A'' B''}$ is a purification of the state in \Cref{eq:nice_form1}. Using Uhlmann's Theorem~\cite{uhlmann1976transition} (also known as the Schr\"{o}dinger-HJW Theorem~\cite{schrodinger1935discussion,hughston1993complete}), there exists an isometry $W$ on $\Hilb_B$ such that
	\[
		(\id \otimes W) \ket{\xi_1}^{RA' A'' B} = \ket{\rho}^{RA'B'} \otimes \ket{\mes_d}^{A'' B''} \;.
	\]
	Since $\ket{\xi_1} = (V \otimes \id) \ket{\tau}$, we have that
	\[
		\rho^{A'B'} \otimes \ketbra{\mes_d}{\mes_d}^{A'' B''} = \Tr_R ((\id \otimes W) \ketbra{\xi_1}{\xi_1} (\id \otimes W)^*) = (V \otimes W) \tau (V \otimes W)^* \;.
	\]
	Since $\tau' = (V \otimes \id) \tau (V \otimes \id)^*$ we obtain Item~2 of \Cref{def:nice-form} for the protocol~$ (\tau',(U_i V^*)) $.
In what follows we let~$\zeta$ denote~$\rho^{A'}$ (which also equals~$\tau^{\prime A^\prime}$).
We now establish Item 3 of \Cref{def:nice-form}. 
\suppress{
\Cref{lem:gauge_freedom} implies that $(\tau',(U_i V^*))$ is also a superdense coding protocol, so by \Cref{lem:orthog} we have that for all $i \neq j$,
	\[
		\Tr_{A''} \left ( U_i V^* \left ( \zeta \otimes \frac{\id}{d} \right ) (U_j V^*) ^* \right) = 0\;.
	\]
}
	Let~$ 
		\sum_k \lambda_k \Pi_k
	$ be the spectral decomposition of $\zeta$ where the $\{ \lambda_k\}$ are distinct and nonzero, and $\Pi_k$ is the orthogonal projector onto the eigenspace of $\zeta$ corresponding to eigenvalue $\lambda_k$. Since by \Cref{eq:nice_form1} we have
	\begin{equation}
	\label{eq:nice_form2}
		U_i V^* \left ( \zeta \otimes \frac{\id}{d}\right) (U_i V^*)^* = \Tr_{BR}(\xi_i) = \zeta_i \otimes \frac{\id}{d}
	\end{equation}
	for some density matrix $\zeta_i$, the states $\zeta_i$ and $\zeta$ have the same eigenvalues with the same multiplicities. That is, there are an orthogonal set of projectors $\{ \Pi_k^{(i)} \}_k$ such that
	\[
		\zeta_i = \sum_k \lambda_k \Pi_k^{(i)} \enspace,
	\]
	where $\dim(\Pi_k^{(i)}) = \dim(\Pi_k)$ for all $i \in [d^2]$. It follows that for all $i$,
	\[
		U_i V^* \left( \Pi_k \otimes \frac{\id}{d} \right) (U_i V^*)^* = \Pi_k^{(i)} \otimes \frac{\id}{d} \enspace.
	\]
	For $i \in [d^2]$ let $C_i$ be a unitary operator on $\Hilb_{A'}$ such that $C_i \Pi_k^{(i)} C_i^* = \Pi_k$ for all $k$. Since $\Pi_k^{(1)} = \Pi_k$, our choice of $C_1 = \id$ suffices. Let $U_i' = (C_i \tensor \id)  U_i V^*$. By \Cref{lem:gauge_freedom}, we have $(\tau', ( U_i') )$ is a superdense coding protocol. Furthermore, \Cref{eq:nice_form2} implies that
	\[
		U_i'  \Big( \sum_k \lambda_k \, \Pi_k \otimes \frac{\id}{d} \Big) (U_i')^* = \sum_k \lambda_k \, \Pi_k \otimes \frac{\id}{d} \enspace,
	\]
	which implies Item 3 of \Cref{def:nice-form}.
	
\Cref{lem:gauge_freedom} implies that $(\tau',(U_i'))$ is also a superdense coding protocol, so by \Cref{lem:orthog} we have that for all $i \neq j$,
	\[
		\Tr_{A''} \left ( U_i' \left ( \zeta \otimes \frac{\id}{d} \right ) (U_j') ^* \right) = 0\;.
	\]
Since~$\Pi_k$ commutes with the $(U_i')$ for all $k$, we have
	\begin{align*}
		0 & = \Tr_{A''} \left ( U_i' \left ( \rho \otimes \frac{\id}{d} \right ) (U_j')^* \right) \\
		&= \sum_k \lambda_k \Tr_{A''} \left ( U_i' \left ( \Pi_k \otimes \frac{\id}{d} \right ) (U_j')^* \right) \\
		& = \frac{1}{d} \sum_k \lambda_k \Tr_{A''} \left ( \left ( \Pi_k \otimes \id \right ) U_i' (U_j')^* \right) \enspace.
	\end{align*}
	Since the $\lambda_k$'s are positive, we have
	\[
	\Tr_{A''} \left ( (\Pi_k \otimes \id) U_i' (U_j')^* (\Pi_k \otimes \id) \right ) = 0
	\]
	for all $k$, and $i \neq j$.
This implies Item 4 of \Cref{def:nice-form}.
	
\end{proof}

\section{Rigidity for two-dimensional superdense coding}
\label{sec:d2}

In this section we prove \Cref{thm:rigidity}, that is, rigidity for $2$-dimensional superdense coding protocols (coding $2$ bits into one qubit, with no error). For the remainder of this section we drop the qualification ``$2$-dimensional'' for brevity.

The proof involves a number of steps. 
First, we invoke \Cref{lem:nice_form}, which states that every superdense coding protocol is locally equivalent to one that has a nice form. Then, we argue that up to local equivalence, in every nice form superdense coding protocol, the encoding operators $(U_i)$ can be block-diagonalized. That is, we can write $U_i = \sum_\ell Q_{i\ell} \otimes R_{i\ell}$ where the $\{Q_{i\ell}\}$ are a set of orthogonal projectors acting on $\Hilb_{A'}$ summing to the identity, and the $\{R_{i\ell}\}$ are a set of Hermitian unitary operators acting on $\Hilb_{A''}$. Next, we argue that (again up to local equivalence) across the different $i$'s, the projectors $\{Q_{i\ell} \}$ can be ``matched up'', and the corresponding operators $R_{i\ell}$ are all pairwise orthogonal. This implies that in fact $\{R_{1 \ell}, R_{2 \ell}, R_{3 \ell}, R_{4\ell} \}$ are unitarily equivalent to the standard Pauli matrices $\{\id, \PauliX, \PauliY, \PauliZ\}$. Using the property that local equivalence is a transitive relation, this concludes the argument. 

\Cref{lem:nice_form} is proven in \Cref{sec-nice-form}. We now proceed to prove the remaining steps in detail.

\subsection{Block-diagonalizing nice form protocols}

In this section, we analyze the structure of the encoding operators in a nice form superdense coding protocol in dimension two. We show that they apply a~$2 \times 2$ unitary operator on register~$A''$, controlled by the state in register~$A'$.
\begin{theorem}
\label{thm:block_diagonal}
Let $(\tau,(U_i))$ be a nice-form protocol. Then there exists a locally-equivalent protocol $(\tau',(U_i'))$ that has a nice form and for $i \in \{2,3,4\}$ we have
	\[
		U_i' =_{\tau'} \sum_\ell Q_{i\ell} \otimes R_{i\ell}\;.
	\]
	for some orthogonal projectors $\{Q_{i\ell}\}_\ell$ on $\Hilb_{A'}$ that sum to $\id$, and $2 \times 2$ traceless, Hermitian unitary matrices $\{R_{i\ell}\}_\ell$ on $\Hilb_{A''}$.

\end{theorem}

\begin{proof}
Fix an $i \in \{2,3,4\}$. Since $(\tau,(U_i))$ has a nice form (see \Cref{def:nice-form}), this means that $\tau^A = \rho^{A'} \otimes \frac{\id}{2}$ for some density matrix $\rho$, and $\Tr_{A''}((\Pi_k \otimes \id)U_i U_1^* (\Pi_k \otimes \id)) =\Tr_{A''}((\Pi_k \otimes \id)U_i (\Pi_k \otimes \id)) = 0$, where $\Pi_k$ is a non-zero eigenspace of $\rho$.

Fix a $k$. By property~\ref{p3} of a nice-form protocol, $U_i$ commutes with $\Pi_k \otimes \id$, and we can write 
\[
U_i \quad = \quad (\Pi_k \tensor \id) U_i (\Pi_k \tensor \id) 
    + (\id - \Pi_k \tensor \id) U_i (\id - \Pi_k \tensor \id) \enspace.  
\]
Let $\hat{U}_{ik} = (\Pi_k \tensor \id) U_i (\Pi_k \tensor \id)$, and note that $\hat{U}_{ik}$ is unitary on the image of $\Pi_k \otimes \id$, i.e.:
\[
	\hat{U}_{ik} \hat{U}_{ik}^* = \hat{U}_{ik}^* \hat{U}_{ik} = \Pi_k \otimes \id \;.
\]
For notational convenience we drop the subscripts $i$ and $k$ until the very end. Let $\hat{U} \eqdef \hat{U}_{ik}$ and let $\Pi \eqdef \Pi_k$.

The condition that $\Tr_{A''}(\hat{U}) = 0$ implies that we can write $\hat{U}$ as
\[
\hat{U}=
\left (
\begin{array}{c|c}
F & G \\
\hline
H & -F
\end{array}
\right ) \enspace,
\]
where $F,G,H$ are block matrices that act on the image of $\Pi$ and the block partitions are with respect to the tensor factor $\Hilb_{A''}$. 

Let 
\[
    F = D_F T_F~, \qquad G = D_G T_G~, \qquad H = D_H T_H
\]
give the polar decompositions of $F,G,H$ respectively where $D_F, D_G,D_H$ are positive semidefinite and $T_F,T_G,T_H$ are unitary on the image of $\Pi$. 

Then
\[
\hat{U} =
\left(
\begin{array}{c|c}
D_F T_F & D_G T_G \\
\hline
D_H T_H & - D_F T_F
\end{array}
\right).
\]
The relation $\hat{U} \hat{U}^* = \Pi \otimes \id$ implies that $D_F^2 = \Pi - D_G^2 = \Pi - D_H^2$. For notational brevity we write $D \eqdef D_F$ and $\wt{D} \eqdef D_G = D_H = \sqrt{\Pi - D^2}$. Note that~$D$ and~$\wt{D}$ have support in the image of~$\Pi$ and are simultaneously diagonalizable.
\suppress{
We can then write
\[
(T_F^* \otimes \id) \hat{U} =
\left(
\begin{array}{c|c}
T_F^* D_F T_F & T_F^*  D_G T_G  \\
\hline
T_F^*  D_H T_H & -T_F^*  D_F T_F
\end{array}
\right).
\]

Next we make use of the unitarity of $(T_F^* \otimes \id) \hat{U}$. The relation $\hat{U} \hat{U}^* = \Pi \otimes \id$ implies
\begin{gather}
	T_F^*  D_F^2 T_F + T_F^*  D_G^2 T_F = \Pi \enspace, \qquad \text{and}
	\label{eq:uud_1} \\
	T_F^*  D_F^2 T_F + T_F^*  D_H^2 T_F = \Pi \;.
	\label{eq:uud_2}	
\end{gather}
\Cref{eq:uud_1,eq:uud_2} imply that $D_F^2 = \Pi - D_G^2 = \Pi - D_H^2$, so for notational brevity we write $D = D_F$ and $\wt{D} = D_G = D_H = \sqrt{\Pi - D^2}$.
}
Write $K \eqdef T_F^*  D T_F $ and $\wt{K} \eqdef T_F^* \wt{D} T_F$. Note that $K$ and $\wt{K}$ are positive semidefinite, and also simultaneously diagonalizable. Write $W_G \eqdef T_F^* T_G$ and $W_H^* \eqdef T_F^* T_H$. Continuing our simplification, we see that
\[
\def\arraystretch{1.5}
(T_F^* \otimes \id) \hat{U}=
\left(
\begin{array}{c|c}
K & \wt{K} W_G \\
\hline
\wt{K} W_H^* & -K
\end{array}
\right).
\]

Now our goal is to find a unitary operator $E$ acting on $\Hilb_{A'}$ such that $(E T_F^* \otimes \id) \hat{U}$ is Hermitian. This is equivalent to the conditions that $EK = KE^*$ and $E\wt{K} W_G = (E \wt{K} W_H^*)^*$. We construct such an operator~$E$ using some relations between the operators~$K, \wt{K}, W_G, W_H$ that we derive below.

We use the unitarity relation $\hat{U}^* \hat{U} = \Pi \tensor \id$ to obtain the equations
\begin{align}
	K^2 + W_G^* \wt{K}^2 W_G & = \Pi \enspace, \qquad \text{and}
	\label{eq:uud_3} \\
	K^2 + W_H \wt{K}^2 W_H^* & = \Pi \enspace.
	\label{eq:uud_4}
\end{align}
These equations, along with the definitions of~$K$ and~$\wt{K}$, imply that the unitary operators $W_G,W_H$ are block-diagonal with respect to the eigenspaces of $K$ and $\wt{K}$ (and therefore commute with $K$ and $\wt{K}$).
Another relation we get from unitarity is $K \wt{K} W_G  = W_H \wt{K} K$, 
which via commutativity of $W_H, \wt{K}, K$ implies
\begin{equation}
\label{eq:ta_tb_2}
K \wt{K} W_G = K \wt{K} W_H\;.
\end{equation}
Let~$\Pi_+$ be the orthogonal projector onto~$\support(K)$. Since $W_G$ commutes with $K,\wt{K}$ (and so does $W_H$), we have that~$\Pi_+$ commutes with~$\wt{K}, W_G, W_H$. By \Cref{eq:ta_tb_2}, we have~$\Pi_+ \wt{K} W_G = \Pi_+ \wt{K} W_H = W_H \wt{K} \, \Pi_+$. Since~$K^2 + \wt{K}^2 = \Pi$, we also have that~$(\Pi - \Pi_+) \wt{K} = (\Pi - \Pi_+)$.
Note that \Cref{eq:ta_tb_2}, together with the fact that $W_H,W_G$ are block-diagonal with respect to the eigenspaces of the product $K \wt{K}$, implies that $W_G$ and $W_H$ must be equal on the support of $K \wt{K}$ (equivalently, the support of $\Pi_+$).

We now construct the desired unitary $E$. 
Let $\Pi_0$ denote $\Pi - \Pi_+$ (i.e., the projection onto the kernel of $K$ within the image of $\Pi$). Left-multiplying both sides of \Cref{eq:ta_tb_2} with $K^{+}$ (the pseudoinverse of $K$ on the image of $\Pi_+$) and right-multiplying both sides by $\Pi_+$ we get
\begin{equation}
\label{eq:pi_plus}
	\Pi_+ \wt{K} W_G \Pi_+ = \Pi_+  \wt{K} W_H  \Pi_+ \;.
\end{equation}
Recall that $\Pi_0 \wt{K} = \Pi_0$; combined with the fact that $W_G$ and $W_H$ are both block-diagonal with respect to the eigenspaces of $K$ and $\wt{K}$, we have
\begin{equation}
\label{eq:pi_all}
	\wt{K} W_G= \Pi_+ \wt{K} W_G \Pi_+ + \Pi_0 W_G \Pi_0  \qquad \text{and} \qquad \wt{K} W_H^* = \Pi_+ \wt{K}  W_H^* \Pi_+ + \Pi_0 W_H^* \Pi_0~.
\end{equation}
Furthermore, $V_G \eqdef \Pi_0 W_G \Pi_0$ and $V_H \eqdef \Pi_0 W_H^* \Pi_0$ are unitary on $\Pi_0$. Let $M \eqdef V_H^* V_G$, and consider its spectral decomposition $M = \sum_j e^{\complexi \theta_j} \ketbra{v_j}{v_j}$ where $\{ \ket{v_j} \}$ is an orthonormal basis for the support of $\Pi_0$. Let $M^{1/2} \eqdef \sum_j e^{\complexi \theta_j/2} \ketbra{v_j}{v_j}$ denote the principal square root of $M$. Define $E_0 \eqdef M^{1/2} V_G^\adjoint$. Observe that $E_0$ is supported only on $\Pi_0$ and satisfies
\begin{equation}
\label{eq:pi_0}
	E_0 V_G = M^{1/2} = (M^{1/2} M^\adjoint)^\adjoint = (E_0 V_H)^\adjoint.
\end{equation}
Consider the operator $E \eqdef \Pi_+ + E_0$ that is unitary on the support of $\Pi$, and acts non-trivially only on the support of $\Pi_0$. Combining Equations~\eqref{eq:pi_plus},~\eqref{eq:pi_all} and~\eqref{eq:pi_0} we get
\begin{align*}
 	E \wt{K} W_G &= \Pi_+ \wt{K} W_G \Pi_+ + E_0 V_G \\
						 &= \Pi_+ \wt{K} W_H \Pi_+ + (E_0 V_H)^*  \\
						 &= ( \Pi_+ W_H^* \wt{K} \Pi_+ + E_0 V_H)^*   \\
						 &= ( \Pi_+  \wt{K} W_H^* \Pi_+ + E_0 V_H)^*  \\
						 &= (E \wt{K} W_H^*)^* \;.
\end{align*}
where the second line follows from \Cref{eq:pi_plus} and \Cref{eq:pi_0}, the fourth line follows because $\wt{K}$ commutes with $W_H^*$, and the last line follows because of \Cref{eq:pi_all}. We also have that $E$ commutes with $\wt{K}$: this is because $\Pi_+$ commutes with $\wt{K}$ and also $E_0$ acts nontrivially only on the eigenspace of $\wt{K}$ with eigenvalue $1$. 

Let $L \eqdef E \wt{K} W_G$. Putting everything together, we have
\begin{equation}
\def\arraystretch{1.5}
\label{eq:u2_form}
(E T_F^* \otimes \id) \hat{U} =
\left(
\begin{array}{c|c}
E K & E \wt{K} W_G \\
\hline
E \wt{K} W_H^* & - E K
\end{array}
\right)
=
\left(
\begin{array}{c|c}
K & L \\
\hline
L^* & - K
\end{array}
\right).
\end{equation}
where we used the fact that $E K = K$. %

Let $K = \sum_r \alpha_r P_r$ and $\wt{K} = \sum_r \sqrt{1 - \alpha_r^2} \, P_r$ be spectral decompositions of $K$ and $\wt{K}$ where the reals $\alpha_r$'s are nonnegative and distinct, and the operators $P_r$ are orthogonal projectors summing to $\Pi$. The operator $\wt{K}$ has such a spectral decomposition because $K^2 + \wt{K}^2 = \Pi$. 
Next, since the unitary operators $E$ and $W_G$ commute with $\wt{K}$, they are block-diagonal with respect to the projectors $\{P_r\}$. Thus 
\begin{align*}
    P_r L P_r &= P_r E \wt{K} W_G P_r = P_r \sqrt{\wt{K}} E W_G \sqrt{\wt{K}} P_r = \sqrt{1 - \alpha_r^2} \,\, P_r E W_G P_r~.
\end{align*}
The operator $P_r E W_G P_r$ is unitary on $P_r$ and we can express it as $\sum_s \beta_{rs} Q_{rs}$ where the $\beta_{rs}$'s are complex numbers on the unit circle and $\{ Q_{rs} \}_s$ are orthogonal projectors that sum to $P_r$. Thus we can write $K = \sum_{r,s} \alpha_r \, Q_{rs}$ and $L = \sum_{r,s} \sqrt{1 - \alpha_r^2} \, \beta_{rs} \,  Q_{rs}$, and $(E T_F^* \otimes \id) \hat{U}$ can be written as
\begin{equation}
	(E T_F^* \otimes \id) \hat{U} = \sum_{r,s} Q_{rs} \otimes R_{rs}
\end{equation}
where $R_{rs}$ is the $2 \times 2$ matrix
\[
\begin{pmatrix}
	\alpha_r	& \sqrt{1 - \alpha_r^2}\cdot \beta_{rs} \\
	\sqrt{1 - \alpha_r^2} \cdot \beta_{rs}^* & -\alpha_r 
\end{pmatrix} \enspace.
\]
Notice that $R_{rs}$ has determinant $-1$ and is traceless, therefore its eigenvalues are $\{+1,-1\}$.

Re-introducing the indices $i \in \{2,3,4\}$ and $k$, we have deduced that for every block of $U_i$ corresponding to the eigenspace $\Pi_k$, there exists a map $S_{ik}$ that is unitary on the image of $\Pi_k$ such that 
\[
(S_{ik} \otimes \id) \hat{U}_{ik} = \sum_{r,s} Q_{ikrs} \otimes R_{ikrs}
\]
where the $(R_{ikrs})_{r,s}$ are $2 \times 2$ Hermitian unitary operators with trace $0$. Define the unitary operator $S_i$ on $\Hilb_{A'}$ as $S_i \eqdef (\id - \sum_k \Pi_k) + \sum_k S_{ik} $.
If we sum over $k$, we get
\[
	(S_i \otimes \id) \hat{U}_i = \sum_\ell Q_{i\ell} \otimes R_{i\ell} \enspace,
\]
where $\hat{U}_i \eqdef \sum_k \hat{U}_{ik}$ and we have re-indexed the sum over $k, r, s$ to be a sum over indices $\ell$.
Let $U_i' \eqdef (S_i \otimes \id)U_i$ for all $i \in \{2,3,4\}$. Then, letting $P \eqdef \sum_k \Pi_k$ denote the projector onto the support of $\rho$, we have
\[
	U_i' \tau (U_i')^* = (S_i U_i P) \, \tau \, (S_i U_i P)^* = (S_i \hat{U}_i P) \, \tau \,(S_i \hat{U}_i P)^* = S_i \hat{U}_i \, \tau \, \hat{U}_i^* S_i^* 
\]
where have suppressed the tensoring with identity that extends all the operators to the same space, and used the property that $(P \otimes \id) \tau (P \otimes \id) = \tau$, $\hat{U}_i P = U_i P$, and $\hat{U}_i P = \hat{U}_i$. Thus, the unitary operators $U_i'$ satisfy the conclusions of the theorem statement. Let $\tau' \eqdef \tau$, so that $(\tau',(U_i'))$ is a superdense coding protocol by \Cref{lem:gauge_freedom}. 

Furthermore, since $(\tau,(U_i))$ has a nice form, it can be verified that $(\tau',(U_i'))$ also has a nice form. First,  since $S_1 = \id$, we have that $U_1' = \id$ (and hence Item 1 of \Cref{def:nice-form} is satisfied). Item 2 of \Cref{def:nice-form} is satisfied since $\tau' = \tau$. Third, since $U_i$ commutes with $\rho \otimes \frac{\id}{2}$ and $S_i$ is block-diagonal with respect to the eigenspaces of $\rho$, it follows that $U_i'$ also commutes with $\rho \otimes \frac{\id}{2}$ (so Item 3 of \Cref{def:nice-form} is satisfied).  Finally, we have \begin{align*}
		\Tr_{A''} (( \Pi_k \otimes \id) U_i' (U_j')^* (\Pi_k \otimes \id)) &= 	\Tr_{A''} ( (S_{ik} \otimes \id) \hat{U}_{ik} \hat{U}_{jk}^* (S_{jk}^* \otimes \id)) \\
		&= S_{ik} \left[ \Tr_{A''} (\hat{U}_{ik} \hat{U}_{jk}^*) \right] S_{jk}^* \\
		&= S_{ik} \left[ \Tr_{A''} (( \Pi_k \otimes \id) U_i U_j^* ( \Pi_k \otimes \id)) \right] S_{jk}^* \\
		&= 0.
	\end{align*}
Thus, Item 4 of \Cref{def:nice-form} is satisfied. This completes the proof of the Theorem.
\end{proof}

\subsection{Matching the blocks of the encoding operators}

In the previous section we saw how, up to local equivalence of protocols, we can express the encoding operator $U_i$ in a two-dimensional superdense coding protocol as a block-diagonal matrix with $2 \times 2$ Hermitian unitary operators on the diagonal. In this section we relate the decompositions to each other. Ultimately, the conclusion is that the blocks ``line up'', so that the operators in the same diagonal block of the four encoding operators are the four single qubit Pauli operators.

\begin{theorem}
\label{thm:matching}
	Let $(\tau,(U_i))$ be a superdense coding protocol that has a nice form where for $i \in \{2,3,4\}$ we have
	\[
		U_i =_\tau \sum_k Q_{ik} \otimes R_{ik}\;,
	\]
	for some orthogonal projectors $\{Q_{ik}\}_k$ on $\Hilb_{A'}$ that sum to $\id$, and $2 \times 2$ traceless, Hermitian unitary matrices $\{R_{ik}\}_k$ on $\Hilb_{A''}$.	Then $(\tau,(U_i))$ is locally equivalent to a superdense coding protocol $(\tau',(U_i'))$ that has a nice form and satisfies the following: there exist orthogonal projectors $\{ K_r \}$ on $\Hilb_{A'}$ that sum to the identity and $2 \times 2$ traceless, Hermitian unitary operators $\{ R_{ir} \}$ such that for all $i \in \{2,3,4\}$,
	\[
		U_i' =_{\tau'} \sum_r K_r \otimes R_{ir}\;.
	\]
	Furthermore, for all $r, i \neq j$ we have
	\[
		\Tr( R_{ir} R_{j r}) = 0 \enspace.
	\]
\end{theorem}
\begin{proof}

	The first step is to ``coarse-grain'' the projectors $\{Q_{ik} \}$ so that the associated operators $R_{ik}$ are all inequivalent in the following sense. For each $i$, we say that $k$ and~$k'$ are \emph{$i$-equivalent\/} if $R_{ik} = \pm R_{ik'}$. For every $i$, this forms an equivalence relation on the $k$'s. Let $p_i(k)$ denote the least $k'$ such that $k'$ and $k$ are $i$-equivalent. Define $s_{ik} \in \{\pm 1\}$ to be such that $R_{ik} = s_{ik} R_{i p_i(k)}$. 
	
	For every $i \in \{2,3,4\}$, for every $k$, define the unitary operator $S_i = \sum_k s_{ik} Q_{ik}$ which acts on $\Hilb_{A'}$ (and set $S_1 = \id$). Then if we define $U'_i = (S_i \otimes \id) U_i$ and $\tau' = \tau$, by \Cref{lem:gauge_freedom} we get that the pair $(\tau',(U_i'))$ is a superdense coding protocol, and furthermore the operators $(U_i')$ admit a block-diagonalization where for all $i$, the associated $2 \times 2$ unitary operators $R_{ik}$ are all inequivalent. It is also straightforward to check that $(\tau',(U_i'))$ has a nice form.

	Next, from Lemma~\ref{lem:nice_form} we have that for $i \neq j$ 
	\begin{equation}
	\label{eq-orth-operators}
		0 = \Tr_{A''} ( U_i' (U_j')^*) = \sum_{k, \ell} Q_{ik} Q_{j\ell } \cdot \Tr(R_{ik} R_{j\ell}).
	\end{equation}
	By left-multiplying the above expression by $Q_{ik}$ for some $k$ and right-multiplying by $Q_{j\ell}$ for some $\ell$ yields $Q_{ik} Q_{j\ell } \cdot \Tr(R_{ik} R_{j\ell}) = 0$. Therefore, if $Q_{ik} Q_{j\ell }\neq 0$, it follows that $\Tr(R_{ik} R_{j\ell}) = 0$.

	Given three sets of projectors $\{ Q_{2k} \}$, $\{Q_{3\ell} \}$, and $\{Q_{4m} \}$ we can define the following tripartite graph $G$, which we call the \emph{overlap graph}. Associate a vertex with every projector $Q_{ik}$ for $i \in \{2,3,4\}$. Include an edge between $Q_{ik}$ and $Q_{j\ell}$ if and only if $Q_{ik} Q_{j\ell } \neq 0$. 	A \emph{triangle} $T=(k,\ell,m)$ in the graph $G$ corresponds to a triple of projectors $Q_{2k}, Q_{3\ell}, Q_{4m}$ such that the pairwise products are all nonzero. Given encoding operators as in \Cref{eq-orth-operators}, we use triangles to match their blocks.
	
	\begin{lemma}[Reduction Lemma]
	\label{lem:reduction}
	Let $\{ Q_{2k} \}$, $\{Q_{3\ell} \}$, and $\{Q_{4m} \}$ be sets of orthogonal projectors with the following properties:
	\begin{enumerate}
		\item $\sum_k Q_{2k} = \sum_\ell Q_{3\ell} = \sum_m Q_{4m}$ 
		\item For $i \neq j$, for all $k,\ell$, $Q_{ik} Q_{j\ell} \neq 0$ implies that $\Tr(R_{ik} R_{j\ell}) = 0$. 
		\item For all $i \in \{2,3,4\}$, the $\{R_{ik} \}_k$ are inequivalent.
	\end{enumerate}
	Then there exists a triangle $T = (k,\ell,m)$ and a unit vector $\ket{v} \in \Hilb_{A'}$ such that 
	\[
		Q_{2k} \ket{v} = Q_{3\ell} \ket{v} = Q_{4m} \ket{v} = \ket{v}.
	\]
	\end{lemma}
	
	We shall assume for now that the Reduction Lemma holds. We show how this gives us an iterative decomposition procedure to construct the orthogonal projectors $\{K_r \}$ satisfying the conclusions of the Theorem.
	
	The sets $Q_2^{(0)} \eqdef \{ Q_{2k} \}$, $Q_3^{(0)} \eqdef \{Q_{3 \ell} \}$, and $Q_4^{(0)} \eqdef \{Q_{4 m} \}$ satisfy the required conditions of the Reduction Lemma with $\sum_k Q_{2k} = \sum_\ell Q_{3\ell} = \sum_m Q_{4m} = \id$. Thus there exists a triangle $T_0 = (k_0,\ell_0,m_0)$ and a vector $\ket{v_0}$ that is a common eigenvector of $Q_{2k_0}, Q_{3\ell_0}, Q_{4m_0}$. Thus we can write
	\[
		Q_{2k_0} = \ketbra{v_0}{v_0} + Q_{2k_0}'	\qquad Q_{3\ell_0} = \ketbra{v_0}{v_0} + Q_{3\ell_0}' \qquad Q_{4m_0} = \ketbra{v_0}{v_0} + Q_{4m_0}'
	\]
	where $Q_{2k_0}'$, $Q_{3\ell_0}'$, and $Q_{4m_0}'$ are orthogonal projectors with rank one smaller.

	Define the sets $Q_2^{(1)}, Q_3^{(1)}, Q_4^{(1)}$ to be the sets $Q_2^{(0)}, Q_3^{(0)}, Q_4^{(0)}$ with the projectors $Q_{2k_0}, Q_{3\ell_0}, Q_{4m_0}$ replaced by $Q_{2k_0}',Q_{3\ell_0}',Q_{4m_0}'$. %
	
	Observe that $Q_2^{(1)}, Q_3^{(1)}, Q_4^{(1)}$ satisfies the required conditions of the Reduction Lemma, with 
	\[
	\sum_{F \in Q_i^{(1)}} F \quad = \quad \id - \ketbra{v_0}{v_0} \enspace,
	\]
	for all~$i \in \set{2, 3, 4}$.
	Applying the Reduction Lemma again, we find another triangle $T_1$ and a common eigenvector $\ket{v_1}$ of the triangle. We continue this process of reducing the rank of at least one operator each in the sets $Q_2^{(r)}, Q_3^{(r)}, Q_4^{(r)}$ and finding common eigenvectors~$\ket{v_r}$ until we have fully expressed
	\[
		U_i' =_\tau \sum_r K_r \otimes R_{ir} \enspace,
	\]
	where $K_r \eqdef \ketbra{v_r}{v_r}$, and for every $r$, the pairwise inner products satisfy $\Tr(R_{2r} R_{3r}) = \Tr(R_{2r} R_{4r}) = \Tr(R_{3r} R_{4r}) = 0$. This concludes the proof of the Theorem. 
\end{proof}

Before proving the Reduction Lemma we establish the following lemma, which claims that up to conjugation by the same unitary operator, the only collection of $2 \times 2$ traceless, Hermitian, unitary, mutually orthogonal matrices are the single-qubit Pauli matrices.

\begin{lemma}
\label{lem:mut_orthog_unitaries}
	Let $R_2,R_3,R_4$ be $2 \times 2$ unitary matrices that are traceless, Hermitian, and satisfy $\Tr(R_i R_j) = 0$ for all $i \neq j$. Then there exists a $2 \times 2$ unitary operator $S$ such that
	\[
		R_2 = S \PauliZ S^*	\qquad R_3 = S \PauliX S^* \qquad R_4 = S \PauliY S^*.
	\]
\end{lemma}

\begin{proof}

We find a sequence of unitary operators $S_1,S_2,S_3$ such that $S \eqdef S_1^* S_2^* S_3^*$ satisfies the conclusions of the lemma. 
	Because $R_2$ is unitary, Hermitian and traceless, we can unitarily diagonalize it as $R_2 = \ketbra{a}{a} - \ketbra{b}{b}$.
Define $S_1$ as the unitary operator with
\[
	S_1 \ket{a} = \ket{0} \qquad S_1 \ket{b} = \ket{1}.
\]
Let $R_i' \eqdef S_1 R_i S_1^*$ for $i \in \{2,3,4\}$. These are all traceless, Hermitian, pairwise orthogonal unitary matrices, and furthermore~$R_2' = \PauliZ$. Suppose
\suppress{
\[
	R_2' = \begin{pmatrix} 1 & 0 \\ 0 & -1 \end{pmatrix} \qquad R_3' = \begin{pmatrix} r & s \\ t & u \end{pmatrix} \;.
\]
}
\[
	R_3' = \begin{pmatrix} r & s \\ t & u \end{pmatrix} \;.
\]
Since $\Tr(R_2' R_3') = 0$ and $R_3'$ is traceless, we have that $r = u = 0$, and since $R_3'$ is Hermitian and unitary, we have $s = t^* = \e^{\complexi \theta}$ for some $\theta \in [0,2\pi)$. Let
\[
S_2 \eqdef \begin{pmatrix} \e^{-\complexi \theta/2} & 0 \\ 0 & \e^{\complexi \theta/2} \end{pmatrix} \enspace,
\]
and $R_i'' \eqdef S_2 R_i' S_2^*$ for $i \in \{2,3,4\}$. Again, the operators~$R_i''$ remain traceless, Hermitian unitary, and pairwise orthogonal, and furthermore~$R_2'' = \PauliZ$ and~$R_3'' = \PauliX$. Suppose
\suppress{
\[
R_2'' = \begin{pmatrix} 1 & 0 \\ 0 & -1 \end{pmatrix} \qquad R_3'' = \begin{pmatrix} 0 & 1 \\ 1 & 0 \end{pmatrix} \qquad R_4'' = \begin{pmatrix} w & x \\ y & z \end{pmatrix} \;.
\]
}
\[
R_4'' = \begin{pmatrix} w & x \\ y & z \end{pmatrix} \;.
\]
From $\Tr(R_2'' R_4'') = 0$, Hermiticity, unitarity, and tracelessness of $R_4''$ we have again that $w = z = 0$ and $x = y^* = \e^{\complexi \phi}$ for some $\phi \in [0,2\pi)$. From $\Tr(R_3'' R_4'') = 0$ we have that $x = -y$, which means that $x = \pm \complexi$. If $x = - \complexi$, then set $S_3 \eqdef \id$. Otherwise, set $S_3 \eqdef -\complexi \, \PauliZ$.
Let $R_i''' \eqdef S_3 R_i'' S_3^*$ for $i \in \{2,3,4\}$. We have that~$R_2''' = \PauliZ$, $R_3''' = \PauliX$, $R_4''' = \PauliY$.
\suppress{
\[
R_2''' = \begin{pmatrix} 1 & 0 \\ 0 & -1 \end{pmatrix} \qquad R_3''' = \begin{pmatrix} 0 & 1 \\ 1 & 0 \end{pmatrix} \qquad R_4''' = \begin{pmatrix} 0 & -i \\ i & 0 \end{pmatrix} \;.
\]
}
Thus, letting $S = S_1^* S_2^* S_3^*$, we obtain the desired conclusion of the lemma.
\end{proof}

We now turn to proving the Reduction Lemma.
	
	\begin{proof}[Proof of \Cref{lem:reduction}] Define $\Pi = \sum_k Q_{2k}$. 

\paragraph{No two triangles in $G$ share an edge.} Suppose we have two triangles corresponding to projectors $(Q_{2k}, Q_{3\ell}, Q_{4m})$ and $(Q_{2k}, Q_{3\ell}, Q_{4m'})$ for some~$k,l,m,m'$. We then have the equations 
\[	
\Tr(R_{2k} R_{3\ell}) = \Tr(R_{2k} R_{4m}) = \Tr(R_{3\ell} R_{4m}) =  \Tr(R_{2k} R_{4m'}) = \Tr(R_{3\ell} R_{4m'}) = 0.
\]
By \Cref{lem:mut_orthog_unitaries}, this implies that there exists a unitary operator $S$ such that
\[
R_{2k} = S \PauliZ S^*	\qquad R_{3\ell} = S \PauliX S^* \qquad R_{4m} = S \PauliY S^*\;.
\]
Therefore we obtain that 
\[
	\Tr(\PauliZ S^* R_{4m'} S) = \Tr(\PauliX S^* R_{4m'} S) = 0 \enspace.
\]
If
\[
S^* R_{4m'} S = \begin{pmatrix} a & b \\ b^* & d \end{pmatrix} \enspace,
\]
the above equation implies that $a = d = 0$ and $b = -b^*$, or equivalently that $b = \pm \complexi$. Thus $S^* R_{4m'} S = \pm \PauliY = \pm S^* R_{4m} S$, or in other words $R_{4m} = \pm R_{4m'}$, contradicting the assumption that $R_{4m}$ and $R_{4m'}$ are inequivalent.

\paragraph{Every vertex is in a triangle.} Consider $Q_{2k}$ for some~$k$. There exists an index $\ell$ such that $Q_{2k} Q_{3\ell} \neq 0$, because the operators $\{Q_{3\ell} \}$ form a resolution of $\Pi$. Since $\{ Q_{4m} \}$ also forms an orthogonal resolution of $\Pi $, we have that 
\[
	0 \neq Q_{2k} Q_{3\ell} = Q_{2k} \left ( \sum_{m} Q_{4m} \right) Q_{3\ell} \enspace.
\]
This implies that there exists an index $m$ such that $Q_{2k} Q_{4m} \neq 0$ and $Q_{4m} Q_{3\ell} \neq 0$.

\paragraph{Finding a common eigenvector of a triangle.} Fix a triangle $T = (k,\ell,m)$. For notational simplicity we shall write $C \eqdef Q_{2k}$, $D \eqdef Q_{3\ell}$, and $E \eqdef Q_{4m}$. First, observe that if $C(\Pi - D)E \neq 0$, then there exists some index~$\ell'$ such that $C Q_{3\ell'} E \neq 0$. This implies that $(k,\ell',m)$ forms a triangle in $G$. But this cannot happen as this triangle would share the edge $(k,m)$ with $T$. Therefore $C(\Pi - D)E = 0$, i.e., $CE = CDE$.
	
By symmetry, we also get that $CD = CED$ and $ED = ECD$.
Thus we have
\[
	0 \neq CE = CDE = CEDE = CECDE = CDECDE.
\]
Since $CDE$ is a product of three projectors, its spectral norm is at most $1$. 
Let $\ket{v}$ be a unit vector realizing the spectral norm of $CDE$, i.e. such that $\| CDE \ket{v} \| = \| CDE \| > 0$. Then
	\begin{align*}
		\| CDE \| = \| CDE \ket{v} \| = \| CDE CDE \ket{v} \| \leq \| CDE CDE \| \leq \|CDE \|^2.
	\end{align*}
The inequality $\| CDE \| \leq \| CDE \|^2$ implies that $\| CDE \| = 1$ (since it is not zero and is at most one). Therefore $\ket{v}$ is a vector such that $C\ket{v} = D\ket{v} = E\ket{v} = \ket{v}$. 
\end{proof}

We now put everything in this section together to prove \Cref{thm:rigidity}, which we restate here for convenience.

\rigidity*

\begin{proof}
	Putting together \Cref{lem:nice_form}, \Cref{thm:block_diagonal}, and \Cref{thm:matching}, we get that all superdense coding protocols $(\tau,(U_i))$ are locally equivalent to one that has a nice form and satisfies the conclusions of \Cref{thm:matching}. Finally, we apply \Cref{lem:mut_orthog_unitaries} to the conclusions of \Cref{thm:matching} to obtain the conclusions of \Cref{thm:rigidity}. 
\end{proof}

\section{Superdense coding and orthogonal unitary bases}
\label{sec-oub}

In this section, we prove that there are multiple non-equivalent 
superdense coding protocols for transmitting~$d^2$ messages 
for~$d \ge 3$, even when no ancilla is used in the encoding 
process, and there is no error in decoding. This implies that
rigidity of superdense coding protocols for~$d \ge 3$ may only
hold in a relaxed form: as we see in this section, rigidity
may hold only up to the choice of an orthogonal unitary basis 
for the space of linear operators~$\linear( \complex^d)$.

\subsection{The connection with unitary bases}
\label{sec-oub-connection}

We draw a connection between superdense coding and bases for the vector
space of~$d \times d$ complex matrices. Although this connection may be inferred 
from \Cref{lem:nice_form}, we give a simple and direct derivation here.

For any integer~$d > 1$, consider a protocol for superdense coding 
of~$d^2$ classical strings using a shared entangled state with local 
dimension~$d$, and a single~$d$-dimensional message.
Assume that the protocol does not use any ancilla in 
the encoding process, and that there is no decoding error. Such a 
protocol necessarily has a simple form, as we describe below.

First, we argue that the initial shared state is maximally entangled.
Bob's state after the message has support in
a~$d^2$-dimensional space. Since there are~$d^2$ strings, and these
are decoded without error, the corresponding states are orthogonal and
pure. So the mixed state of the entire encoded state, corresponding to a uniformly random string, 
is completely mixed. However, the marginal of this state on the register
initially held by Bob is be the same as the marginal for any fixed string.
Thus Bob's share of the initial state is also the~$d$-dimensional
completely mixed state. This implies that the initial shared state
is maximally entangled.

Any maximally entangled state with local dimension~$d$ is of the
form~$(U \tensor V) \ket{\mes_d}$, where~$U,V$ are unitary operators in~$ \unitary(\complex^d)$,
and~$\ket{\mes_d} \coloneqq \tfrac{1}{\sqrt{d}} \sum_{k  = 0}^{d - 1}
\ket{k}\ket{k}$. Therefore, without loss of generality, we may assume that 
Alice and Bob initially share the state~$\ket{\mes_d}$. When the dimension~$d$
is clear from the context, we omit it from the subscript.

Second, since the encoding of any message is pure, Alice's local 
operations satisfy the following properties. 
On input~$i \in [d^2]$, Alice applies a unitary 
operator~$U_i \in \unitary(\complex^d)$ to her share of the 
state~$\ket{\mes}$, and sends the share to the Bob. Since Bob can decode
the input~$i$ with probability~$1$, the states~$(U_i \tensor \id)
\ket{\mes}$ are all orthogonal, i.e., for all distinct~$i,j \in [d^2]$,
we have
\[
     \bra{\mes}  ( U_i^\adjoint U_j \tensor \id) \ket{\mes} \quad = \quad 0 \enspace.
\]
This condition is equivalent to the property that the operators~$U_i$ are mutually
orthogonal with respect to the Hilbert-Schmidt inner product:
\begin{equation}
\label{eq-orth}
    \trace(U_i^\adjoint U_j) \quad = \quad 0 \enspace, \qquad \textrm{ for all }
        i, j \in [d^2], ~ i \neq j \enspace.
\end{equation}
Thus, the operators form an orthogonal unitary basis for the space of
linear operators on~$\complex^d$.

It is straightforward to verify that any such basis
for~$\linear(\complex^d)$ leads to an errorless superdense coding 
protocol for~$d^2$ classical messages. Thus, the study of rigidity of
superdense coding protocols as above is equivalent to the study of
orthogonal unitary bases.

A well-known example of an orthogonal unitary basis in dimension~$d$ 
is generated by the ``clock'' and ``shift'' operators. The elements of 
this basis are also known as the generalized Pauli operators or the 
Heisenberg-Weyl operators.
Let~$\omega_d \coloneqq \exp \left( \tfrac{2 \pi \complexi}{d} \right)$ 
be a primitive~$d$th root of unity.
For~$i,j \in \set{0, 1, \dotsc, d - 1}$, the~$(i,j)$th 
operator~$P_{ij}$ in the basis is defined as~$P_{ij} \coloneqq
\PauliX_d^i \, \PauliZ_d^j$, where~$\PauliX_d \coloneqq \sum_{k = 0}^{d-1}
\ketbra{k+1 \pmod{d}}{k}$ is the shift (or Pauli~X) operator, 
and~$\PauliZ_d \coloneqq \sum_{k = 0}^{d-1} \omega_d^k \ketbra{k}{k}$ 
is the clock (or Pauli~Z) operator.

\subsection{Uniqueness of orthogonal unitary bases}
\label{sec-oub-uniqueness}

Given an orthogonal unitary basis~$B$ for~$\linear(\complex^d)$, we may 
derive other such bases by conjugating elements of~$B$ by a pair of 
unitary operators, and mutliplying each basis element by a potentially 
different complex number of unit modulus. Since this is a rather
straightforward method to derive new bases, we consider the new basis 
to be equivalent to~$B$.
\begin{definition}
Let~$B_1 \coloneqq \set{ U_i : i \in [d^2] }$ be an orthogonal unitary
basis for~$\linear(\complex^d)$. We say that an orthogonal unitary 
basis~$B_2$ \emph{is equivalent to\/}~$B_1$ if there exist unit complex 
numbers~$\alpha_i \in \unitary( \complex)$ and a pair of unitary 
operators~$V, W \in \unitary( \complex^d)$ such that
\begin{equation}
\label{eq-freedom}
    B_2 \quad = \quad \set{ \alpha_i V U_i W : ~~ i \in [d^2]  } \enspace.
\end{equation}
\end{definition}
We may verify that this defines an equivalence relation.

Another way to construct an orthogonal unitary basis is by
taking tensor products of bases in lower dimensions. Suppose~$d$ 
is composite, with~$d = d_1 d_2$ and~$1 < d_1, d_2 < d$, and~$\set{
U_i : i \in [d_1^2] }$ and~$\set{ V_j : j \in [d_2^2] }$ are orthogonal 
unitary bases for~$\linear(\complex^{d_1})$ and~$\linear(\complex^{d_2})$,
respectively. Then
\[
\set{ U_i \tensor V_j : i \in [d_1^2]; ~ j \in [d_2^2] }
\]
is an orthogonal unitary basis for~$\linear(\complex^{d})$. This hints
at the possibility that are bases that are not equivalent to each other 
under operations as in Eq.~(\ref{eq-freedom}).  The following 
proposition confirms this for dimensions which are powers of two.
\begin{proposition}
\label{prop-2tok-dim}
Suppose~$d = 2^k$ for an integer~$k > 1$. Let~$B_1$ be the basis 
for~$\linear(\complex^{d})$ obtained by taking tensor products 
of~$k$ two-dimensional Pauli~X and~Z operators, i.e., 
\[
B_1 \quad \coloneqq \quad \set{ \bigotimes_{i = 1}^k P_i : 
    P_i \in \set{ \id, \PauliX_2, \PauliZ_2, \PauliX_2 \PauliZ_2 } }
    \enspace.
\]
The basis~$B_1$ is \emph{not\/} equivalent to~$B_2$, the~$d$-dimensional 
clock and shift basis.
\end{proposition}
\begin{proof}
The intuition behind the statement is that tensor products of the 
two-dimensional Pauli operators in~$B_1$ all have at most two distinct eigenvalues 
(either~$1$, or~$\pm 1$, or~$\pm \complexi$), whereas some 
of the operators in the clock and shift basis have~$d$ distinct
complex eigenvalues. Due to the freedom available in generating 
equivalent bases, we need additional arguments to formalize this
intuition.

Suppose that~$B_1$ and~$B_2$ \emph{are\/} equivalent and consider unitary 
operators~$V, W \in \unitary( \complex^d )$ which show their equivalence.
Consider the operator~$P \in B_1$ that is mapped to the identity
in~$B_2$ under the equivalence. Let~$\alpha$ be a complex number of unit 
modulus such that~$\alpha V P W = \id$. Then~$V = \alpha^*
W^\adjoint P^\adjoint$. 

Suppose the operator~$Q \in B_1$ is mapped to the clock
operator~$\PauliZ_d \in B_2$, and that~$\PauliZ_d = \beta V Q W$ for some
complex number~$\beta$. Then~$\PauliZ_d = \beta \alpha^* W^\adjoint P^\adjoint
Q W$. The operator on the right hand side has at most two eigenvalues 
(either~$\beta \alpha^*$, or~$\pm \beta \alpha^*$, or~$\pm \complexi \beta \alpha^*$)
as~$P^\adjoint Q$ has eigenvalues~$1$, or~$\pm 1$, or~$\pm \complexi$. 
However, the clock operator~$\PauliZ_d$ has~$d$ distinct eigenvalues, 
the~$d$th complex roots of unity. Since~$d \ge 4$, we get a contradiction, 
and we conclude that~$B_1$ and~$B_2$ are not equivalent.
\end{proof}
\suppress{
A different construction for even dimensions was shown to us by Vern
Paulsen in personal communication.
}

It is then natural to wonder if there is a unique orthogonal unitary
basis in \emph{prime\/} dimensions, up to the equivalence defined above.
In \Cref{sec-oub-construction} we show that even this does not hold, by giving an explicit construction
of a basis in any dimension~$d \ge 5$ that is not equivalent to the clock 
and shift basis.

After our discovery of non-equivalent bases, we learned that the question 
of uniqueness has been studied before by Vollbrecht
and Werner~\cite{VW00-Pauli-operators}. They prove the uniqueness of the
basis consisting of the Pauli operators in dimension two, and state that
the problem of characterizing orthogonal unitary bases in dimensions
larger than two is open. They also give a construction of
``shift-and-multiply'' bases from a collection of~$d$ complex Hadamard
matrices of dimension~$d \times d$ and a~$d \times d$ Latin square. This
construction and non-equivalent bases are discussed in more detail by 
Werner in subsequent work~\cite{Werner01-teleportation-dense-coding}, 
although the notion of equivalence there does not include multiplication by
phases (complex numbers of unit modulus). Werner states without
proof that the existence of non-equivalent bases in dimension at least five
follows from the existence of non-equivalent Hadamard matrices or 
non-equivalent Latin squares, even when the dimension is prime. In 
dimension three, Werner describes how we may construct non-equivalent bases, 
but does not explicitly present them. We present a concrete 
instance of this construction in \Cref{prop-dim-3}. Altogether, we have the following result.

\begin{theorem}
\label{thm-ineq-oub}
For every dimension~$d \ge 3$, there are orthogonal unitary bases that
are not equivalent to the clock and shift basis.
\end{theorem}

The theorem implies that for any~$d \ge 3$, there are non-equivalent 
superdense coding protocols for transmitting~$d^2$ messages, even when
no ancilla is used in the encoding process, and there is no error in 
decoding.

Orthogonal unitary bases have also been studied in the context of quantum
error-correction under the name ``unitary error bases'' 
(see, e.g., Ref.~\cite{MV16-unitary-error-bases} and
the references therein). In addition to the shift-and-multiply
construction, several other methods such as the ``Hadamard
method'' and the ``algebraic method'' have been proposed for their
construction. The ``quantum shift-and-multiply'' method due to Musto 
and Vicary~\cite{MV16-unitary-error-bases} simultaneously
generalizes the shift-and-multiply and Hadamard methods. Musto and
Vicary give examples of orthogonal unitary bases resulting from
this method that are not equivalent to those derived from any of the other
methods mentioned above. However, they give explicit examples only in 
dimension~4. As far as we can tell, earlier explicit constructions, for
example those due to Klappenecker and
R{\"o}tteler~\cite{KR03-unitary-error-bases}, were also for a few small
dimensions.

\subsection{Some useful properties}
\label{sec-oub-prelims}

Here we present two properties that are used in an explicit construction 
leading to Theorem~\ref{thm-ineq-oub}.
The following property of the eigenvalues of the clock and shift 
operators helps in proving non-equivalence to another basis. Recall 
that~$\omega_d \eqdef \exp \left( \tfrac{2 \pi \complexi}{d} \right)$ is
a primitive~$d$th root of unity.
\begin{lemma}
\label{lem-Pauli-eigenvalues}
Let~$d > 1$ be an integer, and let~$a, b \in \set{ 0, 1, \dotsc, d -
1}$.  The eigenvalues of the operator~$\PauliX_d^a \, \PauliZ_d^b$ are 
all of the form
\[
    \omega_d^l \cdot \exp \!\left( \frac{ab (d - 1) \pi \complexi}{d} \right)  \enspace,
\]
for some~$l \in \set{ 0, 1, \dotsc, d - 1}$.
\end{lemma}
\begin{proof}
Since~$\PauliX_d \, \PauliZ_d = \omega^*_d \, \PauliZ_d \PauliX_d$, we
have~$\left( \PauliX_d^a \, \PauliZ_d^b \right)^d = \omega_d^{ab d (d - 1) / 2}
\, \PauliX_d^{ ad} \, \PauliZ_d^{ bd} = \omega_d^{ab d (d - 1) / 2} \, \id$.
So the eigenvalues of~$\PauliX_d^a \, \PauliZ_d^b$ are~$d$th complex
roots of~$\omega_d^{ab d (d - 1) / 2}$, and the lemma follows.
\end{proof}

We also use the following simple number-theoretic property in the
construction of new orthogonal unitary bases.
\begin{lemma}
\label{lem-divisors}
Any integer~$d \ge 5$ has at most~$d - 2$ positive integer divisors. 
\end{lemma}
\begin{proof}
If~$d$ is prime, then it has exactly two positive integer
divisors,~$1,d$, and the lemma holds.

Suppose~$d$ is composite and has
prime factorization~$p_1^{a_1} p_2^{a_2} \dotsb p_k^{a_k}$,
where~$k$ and~$a_1, a_2, \dotsc, a_k$ are positive integers, 
and~$p_1, p_2, \dotsc, p_k$ are distinct prime numbers arranged in
increasing order. The number of positive integer divisors 
of~$d$ equals~$(a_1 + 1) (a_2 + 1) \dotsb (a_k +1)$.

Since~$d$ is composite, either~$k = 1$ and~$a_1 \ge 2$, or~$k \ge 2$.

Suppose~$k = 1$. If~$p_1 \ge 3$, the lemma follows since~$n + 1 \le q^n
- 2$ for all positive integers~$n \ge 2$, for any~$q \ge 3$.
If~$p_1 = 2$, we have~$a_1 \ge 3$ since~$d = p_1^{a_1}
\ge 5$. Since~$n + 1 \le 2^n - 2$ for all integers~$n \ge 3$, the lemma
again follows.

Now suppose~$k \ge 2$. Since~$n + 1 \le q^n$ and~$n + 1 \le r^n - 1$ for
all~$n \ge 1$ whenever~$q \ge 2$ and~$r \ge 3$, the number of
divisors of~$d$ is bounded as
\begin{eqnarray*}
(a_1 + 1) (a_2 + 1) \dotsb (a_k +1)
    & \le & p_1^{a_1} (p_2^{a_2} - 1) \, p_3^{a_3} \dotsb p_k^{a_k} \\
    & \le & p_1^{a_1} p_2^{a_2} p_3^{a_3} \dotsb p_k^{a_k} - p_1^{a_1}
        \\
    & \le & d - 2 \enspace,
\end{eqnarray*}
as claimed.
\end{proof}

\subsection{Explicit constructions}
\label{sec-oub-construction}

We now proceed to describe an explicit construction for all dimensions $d \geq 5$. 
The construction we give has the same form as the shift-and-multiply construction 
due to Vollbrecht and Werner~\cite{VW00-Pauli-operators}. In particular, 
the bases we present correspond to the construction with the Fourier transform over the 
cyclic group of order~$d$ as the Hadamard matrix, and certain Latin squares
that are not equivalent to the one generated by~$\PauliX_d$.

We construct bases that are not equivalent to the clock and shift basis by 
introducing a modification. In particular, we replace the
operators generated by~$\PauliX_d$ by another sequence. Note that the
operators~$\PauliX_d^i$ correspond to permutations on~$\integers/ d
\integers$, i.e., they permute the standard basis of~$\complex^d$. Conversely 
any permutation~$P$ on~$\integers/ d \integers$ corresponds to the
operator~$\sum_{a \in \integers/ d \integers} \ketbra{P(a)}{a}$ which
permutes standard basis elements of~$\complex^d$. So this is a bijection.

It is also helpful to view a permutation~$P$ on~$\integers/ d
\integers$ as a perfect matching in the complete bipartite 
graph~$K_{d,d}$, with vertex~$a$ in one part being matched with the
vertex~$P(a)$ in the other. This mapping defines a
bijection between permutations and perfect matchings. The construction
we give relies on these three equivalent views of a permutation, and uses 
the same letter to refer to the corresponding matching and the linear 
operator on~$\complex^d$.

We start with the following observation.
\begin{lemma}
Let~$P_0, P_1, \dotsc, P_{d-1}$ be a sequence of~$d$ \emph{disjoint\/}
matchings in the graph~$K_{d,d}$. Then the matrices
\[
\set{ P_i \PauliZ_d^j : 0 \le i,j < d }
\]
form an orthogonal unitary basis.
\end{lemma}
\begin{proof}
Since~$P_i$ permutes the standard basis vectors of~$\complex^d$, it is a
unitary operator. Therefore~$ P_i \PauliZ_d^j$ is also unitary.

Now consider the inner product of~$ P_i \PauliZ_d^j $ and~$P_k \PauliZ_d^l $ for
two pairs~$(i,j)$ and~$ (k,l)$.
We have
\[
\trace( \PauliZ_d^{-j} P_i^{-1} P_k \PauliZ_d^l ) 
    \quad = \quad \sum_{m = 0}^{d - 1} \omega_d^{(l-j)m} \braket{ P_i(m) }{ P_k(m) } \enspace.
\]
If~$i \neq k$, the inner product is~$0$ since for any~$m$, 
the disjoint matchings~$P_i$ and~$P_k$ match the vertex~$m$ to distinct 
vertices. If~$i = k$, but~$l \neq j$, the 
inner product is again~$0$ as~$\omega_d$ is a~$d$th root of unity.
\end{proof}

We show that to derive non-equivalent bases, it suffices to 
have two of the matchings satisfy simple properties.
\begin{lemma}
\label{lemma-new-oub}
Let~$d \ge 2$, and~$P_0, P_1, \dotsc, P_{d-1}$ be a sequence of~$d$ \emph{disjoint\/}
matchings in~$K_{d,d}$ such that
\begin{enumerate}
\item
$P_0$ is the identity permutation, i.e., matches vertex~$i$ in one part
to vertex~$i$ in the other part; and
\item
the permutation~$P_1$ has a cycle of length~$k$ such that~$k$ does not
divide~$d$.
\end{enumerate}
Then the basis~$B \coloneqq \set{ P_i \PauliZ_d^j : 0 \le i,j < d }$
is not equivalent to the clock and shift basis.
\end{lemma}
\begin{proof}
The intuition here is the following. The operator corresponding to 
the permutation~$P_1$ has the~$k$ distinct~$k$th roots of unity as
eigenvalues. In particular, the eigenvalues include~$1$ and~$\omega_k$.
On the other hand, the eigenvalues of any operator in the clock and 
shift basis are of the form~$\gamma \omega_d^l$ for some integer~$l$, 
and a fixed unit complex number~$\gamma \in \unitary(\complex)$
depending only on the operator. Since~$k$ does not divide~$d$, the 
operator~$P_1$ does not belong to the clock and shift basis.

Formally,
suppose that~$B$ is equivalent to the clock and shift basis, and the
equivalence is given by unitary operators~$U,W$. Suppose that the
identity operator~$P_0 \in B$ is mapped to~$\PauliX_d^i \, \PauliZ_d^j$ 
and~$P_1$ is mapped to~$\PauliX_d^k \, \PauliZ_d^l$ under this 
equivalence. That is,~$\PauliX_d^i \, \PauliZ_d^j = \alpha V P_0 W 
= \alpha V W$ and~$\PauliX_d^k
\, \PauliZ_d^l = \beta \, V P_1 W$ for some~$\alpha, \beta \in
\unitary(\complex)$.

From the equation for~$P_0$ we have~$V = \alpha^* \, \PauliX_d^i
\, \PauliZ_d^j W^\adjoint$, so that~$ \PauliX_d^k \, \PauliZ_d^l = 
\beta \alpha^* \, \PauliX_d^i \, \PauliZ_d^j W^\adjoint P_1 W$.
Equivalently, we have
\begin{equation}
\label{eq-contra}
\alpha \beta^* \omega_d^m \, \PauliX_d^{k - i} \, \PauliZ_d^{l - j} 
    \quad = \quad W^\adjoint P_1 W \enspace,
\end{equation}
where~$m = -(k-i) j$.
By Lemma~\ref{lem-Pauli-eigenvalues}, there is a fixed~$\gamma 
\in \unitary(\complex)$ such that the eigenvalues of the 
operator on the left hand side of Eq.~(\ref{eq-contra}) are of the 
form~$\gamma \omega_d^l$ for some integer~$l$.
On the other hand, the operator on the right hand side of the equation 
is similar to~$P_1$. Since~$P_1$ has eigenvalues~$1$ and~$ \omega_k$, 
we have~$1 = \gamma \omega_d^m$ and~$\omega_k = \gamma \omega_d^n$ for
some integers~$m,n$. Eliminating~$\gamma$, we get~$\omega_k =
\omega_d^{n-m}$. This implies that
\[
    \frac{2 \pi \complexi}{ k}
        \quad = \quad \frac{2 \pi \complexi (n - m) }{ d}
            + 2 \pi \complexi p \enspace,
\]
for some integer~$p$, or equivalently that~$d = (n - m + pd) k$. This is
a contradiction, as~$k$ does not divide~$d$.
\end{proof}

Finally, we prove that matchings as in the hypothesis of
Lemma~\ref{lemma-new-oub} exist.
\begin{lemma}
\label{lemma-matching-existence}
For any integer~$d \ge 5$,
there is a sequence of~$d$ \emph{disjoint\/} matchings~$P_0, P_1,
\dotsc, P_{d-1}$ in~$K_{d,d}$ such that
\begin{enumerate}
\item
$P_0$ is the identity permutation, i.e., matches vertex~$i$ in
one part with vertex~$i$ in the other part; and
\item
the permutation~$P_1$ has a cycle of length~$k$ such that~$k$ does not
divide~$d$.
\end{enumerate}
\end{lemma}
\begin{proof}
By Lemma~\ref{lem-divisors}, for any~$d \ge 5$, there is an integer~$k
\in [2, d - 2]$ that does not divide~$d$.
\comment{
Note that~$1,d$ are divisors.
Moreover~$d - 1$ is not a divisor of~$d$ for~$d \ge 3$. If it were, we
would have~$d = (d - 1)k$ with~$k \ge 2$, so~$d = 1 + 1/(k-1) \le 2$.
}

Let~$P_0$ be the identity permutation, and let~$P_1 \coloneqq (0, 1,
\dotsc, k - 1) (k, k + 1, \dotsc, d - 1)$ be a permutation consisting 
of two cycles of length~$k$ and~$d - k$, respectively.
The perfect matchings corresponding to~$P_0$ and~$P_1$ are disjoint,
as~$P_0$ maps each element to itself while~$P_1$ cyclically shifts every
element within each of its two cycles (both of which are of length at 
least two).

Consider the graph~$G$ obtained by deleting the edges in the 
matchings~$P_0$ and~$P_1$ from~$K_{d,d}$. The graph~$G$ is 
a~$(d-2)$-regular bipartite graph. Thus, by the Hall theorem~\cite{Hall35-matching},
$G$ can be decomposed into~$(d - 2)$ disjoint perfect matchings.
\end{proof}

Lemma~\ref{lemma-matching-existence} and Lemma~\ref{lemma-new-oub} 
together imply that for any dimension~$d \ge 5$,
there are multiple non-equivalent orthogonal unitary bases. The same property 
holds for~$d = 4$ due to \Cref{prop-2tok-dim}, and for~$d = 3$ due to 
\Cref{prop-dim-3} below.
This proves Theorem~\ref{thm-ineq-oub}. 

\begin{proposition}
\label{prop-dim-3}
There are orthogonal unitary bases for~$\linear( \complex^3)$ that
are not equivalent to the clock and shift basis.
\end{proposition}
\begin{proof}
Denote the clock and shift basis by~$B$. Note that~$B$ is a commutative projective group under operator composition, i.e., it is closed under taking products of the operators and the operators commute, all up to some phase (a unit complex number) that may depend on the operators. We construct a basis~$B'$ such that the equivalence of~$B$ and~$B'$ implies that~$B'$ also is a commutative projective group. However, the basis~$B'$ has elements that do not commute even up to a phase, which is a contradiction.

We construct~$B'$ following an idea due to Werner~\cite{Werner01-teleportation-dense-coding}; see the discussion after Proposition~9 in the paper. Let~$ M \eqdef \beta \ketbra{0}{0} + \ketbra{1}{1} + \ketbra{2}{2} $, where~$\beta \in \unitary( \complex)$ is a unit complex number such that~$\beta \neq 1$. Let~$B' \eqdef \set{ U_{ij} : 0 \le i,j \le 2}$, where for any~$i \in \set{0, 1, 2}$
\begin{align*}
U_{ij} \quad \eqdef & \quad 
    \begin{cases}
        \PauliX_3^j \, \PauliZ_3^i & \quad j \in \set{ 0, 2} \enspace, \quad \text{and} \\
        \PauliX_3 \, \PauliZ_3^i M & \quad j = 1 \enspace.
    \end{cases}
\end{align*}
We may verify that this is an orthogonal unitary basis for any choice of~$\beta \in \unitary( \complex)$.

Suppose the basis~$B'$ is equivalent to~$B$, and the equivalence is given by the operators~$V,W \in \unitary( \complex^3)$ and unit complex numbers~$\alpha_{ij} \in \unitary( \complex)$. Consider the element~$\PauliX_3^a \, \PauliZ_3^b $ of~$B$ that corresponds to the operator~$U_{00} \in B'$. We have~$\PauliX_3^a \, \PauliZ_3^b = \alpha_{00} V U_{00} W =  \alpha_{00} V W$. Then~$W = \alpha_{00}^* V^\adjoint \PauliX_3^a \PauliZ_3^b $, and
\[
B \quad = \quad \set{ \alpha_{ij} \, \alpha_{00}^* V U_{ij} V^\adjoint \PauliX_3^a \, \PauliZ_3^b ~:~ 0 \le i,j \le 2} \enspace.
\]
Since~$B$ is closed under right multiplication by~$( \PauliX_3^a \,\PauliZ_3^b )^\adjoint $ up to phases, the set of operators
\[
\set{ V U_{ij} M^\adjoint V^\adjoint : 0 \le i,j \le 2}
\]
is also a commutative projective group, as is the basis~$B'$.

We show next that not all operators in the set~$B'$ commute, even up to a phase. Consider the operators~$U_{01}$ and~$U_{02}$. These operators commute up to a phase if and only if there is a unit complex number~$\gamma$ such that
\[
\begin{array}{rrl}
    & U_{01} U_{02} \quad = \quad & \gamma \, U_{02} U_{01} \\
    \Longleftrightarrow \qquad & \PauliX_3 M \, \PauliX_3^2 \quad = \quad & \gamma \, \PauliX_3^2 \, \PauliX_3 M \\
    \Longleftrightarrow \qquad & \ketbra{0}{0} + \beta \ketbra{1}{1} + \ketbra{2}{2}
        \quad = \quad & \gamma \left( \beta \ketbra{0}{0} + \ketbra{1}{1} + \ketbra{2}{2} \right) \enspace.
\end{array}
\]
This implies that~$\gamma = \beta = 1$. As we chose~$\beta \neq 1$, this is a contradiction, and~$B$ and~$B'$ are not equivalent.
\end{proof} 
\section{Random superdense coding protocols}

In this section we study a random protocol for approximate superdense coding. Its analysis draws heavily on  results in high-dimensional probability. We present these results in Section~\ref{sec-random-matrix-theory} and develop some properties of random entangled vectors in Section~\ref{sec:random_mes}, before proceeding to the analysis in Section~\ref{sec-random-protocol-analysis}. Finally, in Section~\ref{sec-expectation-limit}, we address a subtle issue that we encounter in the analysis.

\subsection{Background from random matrix theory}
\label{sec-random-matrix-theory}

In this section, we present some useful results from random matrix theory.

\begin{definition}[Isotropic vector]
We say a random vector~$\ket{\bxi} \in \complex^n$ is \emph{isotropic\/}
if~$\expct \density{\bxi} = \id$.
\end{definition}

Random variables which have tails that decay as fast as the normal distribution play an important role in high dimensional probability. Let~$S^{n-1}$ denote the set of unit vectors in~$\complex^n$.
\begin{definition}[Sub-gaussian random variables and vectors]
A random variable~$\bx \in \complex$ is \emph{sub-gaussian\/} if there exists a parameter~$\kappa > 0$ such that 
\[
	\Pr ( |\bx| \geq t ) \quad \leq \quad 2 \exp \! \big( -t^2/ \kappa^2 \big)
\]
for all $t \geq 0$. The \emph{sub-gaussian norm\/} of~$\bx$, denoted by $\| \bx \|_{\sg}$, is defined as 
\[
	\| \bx \|_{\sg} \quad \eqdef \quad \inf \set{ t > 0 : \expct \exp ( \Abs{ \bx}^2 / t^2) \leq 2 } \enspace.
\]
A random vector $\ket{\bv} \in \complex^n$ is \emph{sub-gaussian\/} if for all unit vectors $\ket{u} \in S^{n-1}$, the inner product $\braket{ u}{ \bv}$ is sub-gaussian. The sub-gaussian norm of $\ket{\bv}$ is defined as
\[
	\norm{ \bv }_{\sg} \quad \eqdef \quad \sup_{u \in S^{n-1}} \norm{ \braket{ u}{ \bv} }_{\sg} \enspace.
\]
\end{definition}

Sub-gaussian norm can be characterized in multiple ways. The following lemma describes two of them; see~\cite[Proposition~2.5.2, Section~2.5.1]{Vershynin18-high-dim-probability}.
\begin{lemma}%
\label{lem-equiv-sub-gaussian}
There are positive universal constants~$c_1, c_2$ such that for any random variable~$\bx \in \complex$,
\begin{enumerate}
	\item If for some parameter~$\kappa_1 > 0$, 
	\[ 
	    \Pr ( |\bx| \geq t) \quad \leq \quad 2 \exp \! \big( - t^2/ \kappa_1^2 \big)
	\]
	for all~$t \geq 0$, then~$ \bx$ is sub-gaussian and~$\norm{ \bx}_\sg \le c_1 \kappa_1$.
	\item If~$ \bx$ is sub-gaussian with~$ \norm{ \bx}_\sg  \le \kappa_2$ for some parameter~$\kappa_2 > 0$, then
	\[
	    \Pr ( |\bx| \geq t) \quad \leq \quad 2 \exp \! \big( - t^2/ c_2 \kappa_2^2 \big)
	\]
	for all~$t \geq 0$.
\end{enumerate}
\suppress{
The following hold for a sub-gaussian random variable $\bx$ with norm $\norm{ \bx}_{\sg}$:
\begin{enumerate}
	\item $\Pr ( |X| \geq t) \leq 2 \exp( -ct^2/ \norm{X}_{\sg}^2)$ for all $t \geq 0$
	\item If $\expct X = 0$ then $\expct \exp( \lambda X) \leq \exp( C\lambda^2 \norm{X}_{\sg}^2)$ for all $\lambda \in \R$.
\end{enumerate}
}
\end{lemma}

\suppress{
\begin{lemma}[Sum of sub-gaussian random variables]
\label{lem-sums-of-sub-gaussians}
Let $X_1,X_2$ be sub-gaussian random variables with norms $\|X_1\|_{\sg}$ and $\|X_2\|_{\sg}$, respectively. Then $X_1 + X_2$ is sub-gaussian with norm at most $\|X_1\|_{\sg} + \|X_2\|_{\sg}$. 
\end{lemma}

\begin{theorem}[Azuma's inequality for sums of sub-gaussian variables]
\label{thm-azuma}
Let $\{ \cal{F}_i \}$ be any filtration, and let $D_1,\ldots,D_n$ be random variables satisfying
\begin{enumerate}
	\item $D_i$ is $\cal{F}_i$-measurable and $\expct[D_i | \cal{F}_{i-1}] = 0$
	\item $\expct[ \exp(\lambda D_i) \mid \cal{F}_{i-1} ] \leq \exp(\lambda^2 \sigma_i^2/2)$ almost surely.
\end{enumerate}
Then the sum $\sum_i D_i$ is sub-gaussian with norm $\sqrt{\sum_i \sigma_i^2}$.
\end{theorem}

\begin{lemma}[Sums of independent sub-gaussian random variables]
\label{lem-sums-of-indp-sub-gaussians}
Let $X_1,\ldots,X_n$ denote independent, mean zero, sub-gaussian random variables. Then $\sum_i X_i$ is also a sub-gaussian random variable, and
\[
	\norm{ \sum_{i=1}^n X_i }^2_{\sg} \leq C \sum_{i=1}^n \norm{ X_i }^2_{\sg}
\]
for some universal constant $C > 0$.
\end{lemma}

\begin{lemma}[Sub-gaussianity of unit vectors]
\label{lem-sub-gaussian-unit-vector}
	Let $\ket{u}$ be a Haar-random unit vector in $\R^d$. Then $\ket{u}$ is sub-gaussian, with norm at most
	\[
		\norm{\ket{u}}_{\sg} \leq \frac{C}{\sqrt{d}}
	\]
	for some universal constant $C > 0$.
\end{lemma}
}

The following theorem gives a sharp bound on the largest singular value of a class of random matrices; see the text by Verhsynin~\cite{Vershynin18-high-dim-probability} for this and related results. Vershynin states the result for real matrices, but the proof extends to complex matrices in a straightforward manner.
\begin{theorem}[\cite{Vershynin18-high-dim-probability}, Theorem 4.6.1]
\label{thm-sub-gaussian}
Let~$\bA \eqdef \sum_{i = 1}^m \ketbra{i}{\bx_i}$ be a complex~$m \times n$ matrix whose rows $\ket{ \bx_i}$ are independent, mean zero, sub-gaussian isotropic vectors in~$\complex^n$. Then there is a universal constant~$c > 0$ such that for all $t \geq 0$, we have
\[
	\| \bA \| \quad \leq \quad \sqrt{m} + c \kappa^2 ( \sqrt{n} + t)
\]
with probability at least~$1 - 2\exp(-t^2)$, where~$ \kappa \eqdef \max_i \| \bx_i \|_{\sg}$.
\end{theorem}

Let~$\| \cdot \|_2$ denote the Hilbert-Schmidt norm on~$\linear( \complex^d)$:
\[
    \| A \|_2 \quad \eqdef \quad \sqrt{ \Tr(A^\adjoint A)} \enspace.
\]
This norm induces the following~$\ell_2$-sum metric on~$ \big( \unitary( \complex^d) \big)^m$:
\[
    \norm{ (U_1, U_2, \dotsc, U_m) - (V_1, V_2, \dotsc, V_m) }_2 
        \quad \eqdef \quad \left( \sum_{i = 1}^m \norm{ U_i - V_i }_2^2 \right)^{1/2} \enspace.
\]
Let~$f : \big( \unitary( \complex^d) \big)^m \rightarrow \reals$ be a continuous function.
We say~$f$ is $\kappa$-Lipschitz with respect to the~$\ell_2$-sum of Hilbert-Schmidt metrics if for all~$(U_i), (V_i) \in \big( \unitary(\complex^d) \big)^m $,
we have
\[
\size{ f(U_1, U_2, \dotsc, U_m) - f(V_1, V_2, \dotsc, V_m) } 
    \quad \le \quad \kappa \norm{ (U_1, U_2, \dotsc, U_m) - (V_1, V_2, \dotsc, V_m) }_2 \enspace.
\]
Let~$\bU_i \in \unitary(\complex^d)$, $1 \le i \le m$ be i.i.d.\ Haar-random unitary operators. If~$\kappa$ is sufficiently smaller than the dimension~$d$, with high probability, the random variable~$f(\bU_1, \bU_2, \dotsc, \bU_m)$ is close to its expectation.  This concentration of measure
property is formalized by the following theorem, which is a special case
of Theorem~5.17 in the book on random matrix theory by Meckes~\cite{M19-random-matrix-theory}.
\begin{theorem}[\cite{M19-random-matrix-theory}, Theorem~5.17, page~159]
\label{thm-Haar-concentration}
Let~$ \bU_i \in \unitary(\complex^d) $, $i \in [m]$, be i.i.d.\ random unitary operators chosen according to the Haar measure. Suppose the function~$f : \big( \unitary(\complex^d) \big)^m \rightarrow \reals$ is~$\kappa$-Lipschitz with respect to the~$\ell_2$-sum of Hilbert-Schmidt metrics, with~$\kappa > 0$. 
Then for every positive real number~$t$, we have
\[
\Pr \! \left( f(\bU_1, \bU_2, \dotsc, \bU_m) \ge \varphi + t \right)
    \quad \leq \quad \exp \! \left( - \frac{ (d-2) t^{2} }{ 24 \kappa^2 } \right) \enspace,
\]
where~$\varphi \eqdef \expct f(\bU_1, \bU_2, \dotsc, \bU_m)$.
\end{theorem}

The Mar\v{c}enko–Pastur theorem characterises the spectrum of a wide class of random matrices in the limit of large dimension. We rely on a version of the theorem due to Yaskov~\cite{Yaskov16-MP-short-proof} that applies to matrices whose entries need not all be independent. While Yaskov states the result for \emph{real\/} matrices, the proof extends to \emph{complex\/} matrices with straightforward modifications. We sketch the observations and the modifications which enable this extension after the statement of the theorem.

The columns of the random matrices we consider satisfy a certain asymptotic isotropy condition.
\begin{definition}
\label{def-pseudo-isotropy}
Let~$m(n)$ be a sequence of positive integers such that~$m \to \infty$ as~$n \to \infty$.
Let~$(\ket{\bx_m})$ be a sequence of random vectors with~$\ket{\bx_m} \in \complex^m$. We say that the sequence~$(\ket{\bx_m})$ is pseudo-isotropic if for all sequences of complex matrices~$(A_m )$ with~$A_m \in \complex^{m \times m}$ and with uniformly bounded spectral norm (i.e., $\|A_m \| \leq \kappa$ for all~$m$ for a universal constant~$\kappa$),
\[
    \frac{1}{m} \Paren{ \bra{\bx_m} A_m \ket{\bx_m} - \Tr(A_m) }  \overset{\rP}{\longrightarrow} 0
\]
as~$m \to \infty$.
\end{definition}
Define the \emph{empirical spectral distribution\/}~(ESD) of an~$m \times m$ positive semi-definite matrix~$A$ as~$\tfrac{1}{m} \sum_{i = 1}^m \diracdelta(x - \lambda_i)$, where~$(\lambda_i : i \in [m])$ are the eigenvalues of~$A$, and~$\diracdelta$ is the Dirac-delta function. This is the probability density function of a uniformly random eigenvalue of~$A$.
\begin{theorem}[Mar\v{c}enko-Pastur law~\cite{Yaskov16-MP-short-proof}]
\label{thm:marcenko-pastur}
Fix an $r > 0$, and let~$m,n$ be integers with~$n, m \ge 1$ and~$m$ a function of~$n$ such that~$m/n \to r$ as~$n \to \infty$. For each~$m$, let~$\ket{\bx_m}$ in~$\complex^m$ be a random vector such that the sequence of vectors~$( \ket{\bx_m})$ is pseudo-isotropic.
Let~$( \bM_{n,m} )$ be a sequence of~$m \times n$ random matrices whose columns are i.i.d.\ copies of the random vector~$\ket{\bx_m}$, and let~$\bmu_{n,m}$ be the ESD of the matrix $\frac{1}{n} \bM_{n,m} \bM_{n,m}^\adjoint$.
Then, as~$n \to \infty$, the ESD~$\bmu_{n,m}$ converges weakly to the density~$\MP_r$ almost surely, where
\[
	\MP_r(x) \quad \eqdef \quad \max \set{ 0, 1 - 1/r } \diracdelta(x)
	    + \frac{ \sqrt{(x - a)(b - x)}}{2\pi r x} \; \Ind(a \le x \le b) \enspace,
\]
with~$a \eqdef (1 - \sqrt{r} \,)^2$, $b \eqdef (1 + \sqrt{r} \,)^2$.
\end{theorem}
In other words, as~$n \to \infty$, with probability~$1$, the cumulative distribution function of a uniformly random eigenvalue of the matrix~$\frac{1}{n} \bM_{n,m} \bM_{n,m}^\adjoint$ converges point-wise to that given by the probability density function~$p_r$.

Theorem~\ref{thm:marcenko-pastur} follows from the proof of Theorem~2.1 in Ref.~\cite{Yaskov16-MP-short-proof} by noting the following points. The eigenvalues of the matrix~$\frac{1}{n} \bM_{n,m} \bM_{n,m}^\adjoint$ are all real, and therefore the Stieltjes continuity theorem applies to~$\bmu_{n,m}$. Further, the Sherman-Morrison formula also extends to the sum~$A + \ketbra{u}{v}$, where~$A$ is an invertible~$m \times m$ complex matrix, and~$\ket{u}, \ket{v}$ are in~$\complex^m$: the matrix~$A + \ketbra{u}{v}$ is invertible if and only if~$1 + \bra{v} A^{-1} \ket{u} \neq 0$, and if the latter condition holds,
\[
(A + \ketbra{u}{v})^{-1} \quad = \quad A^{-1} - \frac{ A^{-1} \ketbra{u}{v} A^{-1} }{ 1 + \bra{v} A^{-1} \ket{u} } \enspace.
\]
We can prove that the Stieltjes transform~$\bs_n(z)$ of~$\bmu_{n,m}$ tends to its expectation~$\expct \bs_n(z)$ almost surely as~$n \to \infty$, following Step~1 in the proof of Theorem~1.1 in Ref.~\cite{BZ08-covariance-matrices}.
The rest of the proof in Ref.~\cite{Yaskov16-MP-short-proof} now extends to the case of interest to us by replacing all instances of the transpose of a real vector by the conjugate transpose of the corresponding complex vector. 

\subsection{Pseudo-isotropy of random maximally entangled vectors}
\label{sec:random_mes}

In this section, we develop properties of linear operators with certain symmetries, and use these to prove that a sequence of random maximally entangled vectors is pseudo-isotropic. This property is later used in the analysis of a random superdense coding protocol.

We consider operators on~$\complex^d \tensor \complex^d \tensor \complex^d \tensor \complex^d$, and label the four tensor factors with~$A,B,C,D$, respectively. As is the convention in quantum information, we use superscripts to indicate the tensor factors on which an operator acts. Let~$\Swap \eqdef \sum_{i,j = 1}^d \ketbra{i,j}{j,i}$ be the \emph{swap\/} operator on~$\complex^d \tensor \complex^d$; it permutes the two tensor factors.
\begin{lemma}
\label{lem-sum-perm}
Let~$W \in \linear(\complex^d \tensor \complex^d \tensor \complex^d \tensor \complex^d )$. Suppose~$W$ commutes with~$\id^{AB} \tensor U^C \tensor U^D$ for all unitary operators~$U \in \unitary( \complex^d)$, as well as with~$U^A \tensor U^B \tensor \id^{CD}$. Then~$W$ is a linear combination of operators of the form~$P^{AB} \tensor Q^{CD}$, where~$P,Q \in \set{\id, \Swap}$.
\end{lemma}
\begin{proof}
Let~$(E_i)$ be a basis for the vector space~$\linear( \complex^d \tensor \complex^d)$. We may express~$W$ as~$W = \sum_{i = 1}^{d^2} E_i \tensor W_i$ for some operators~$W_i \in \linear( \complex^d \tensor \complex^d)$. Since~$W$ commutes with~$\id^{AB} \tensor U^C \tensor U^D$, we have
\[
\sum_{i = 1}^{d^2} E_i \tensor W_i \quad = \quad \sum_{i = 1}^{d^2} E_i \tensor (U \tensor U) W_i (U^\adjoint \tensor U^\adjoint) \enspace,
\]
for all~$U \in \unitary( \complex^d)$. Since the operators~$E_i$ form a basis, we conclude that
\[
W_i \quad = \quad (U \tensor U) W_i (U^\adjoint \tensor U^\adjoint) \enspace,
\]
i.e., the operator~$W_i$ commutes with~$U \tensor U$ for every~$i$. As a consequence of the von Neumann double commutant theorem~\cite[Theorem~7.15, Section~7.1]{Watrous18-TQI}, each operator~$W_i$ may be written as a linear combination of~$\set{ \id, \Swap}$. So
\[
W \quad = \quad \sum_{i = 1}^{d^2} E_i \tensor (\alpha_i \id + \beta_i \Swap)  \enspace,
\]
for some complex numbers~$\alpha_i, \beta_i$. Rearranging the sum, we get that
\[
W \quad = \quad G \tensor \id + H \tensor \Swap  \enspace,
\]
for some operators~$G,H \in \linear( \complex^d \tensor \complex^d)$. Since~$W$ commutes with~$U^A \tensor U^B \tensor \id^{CD}$ as well, and~$\id$ and~$\Swap$ are linearly independent, by~\cite[Theorem~7.15, Section~7.1]{Watrous18-TQI} we similarly get that~$G$ and~$H$ are also linear combinations of~$\set{ \id, \Swap}$. The lemma follows.
\end{proof}

Consider the random vector~$\ket{\bpsi}$ defined as~$\ket{\bpsi} \eqdef (\bU \tensor \id) \ket{\mes_d}$, where~$\bU \in \unitary(\complex^d)$ is a Haar-random unitary operator and~$\ket{\mes_d}$ is the maximally entangled state~$\tfrac{1}{\sqrt{d}} \sum_{k = 1}^d \ket{k}\ket{k}$ with local dimension~$d$. We would like to compute a closed form expression for the operator~$M$ on~$\C^d \otimes \C^d \otimes \C^d \otimes \C^d$ defined as:
\begin{equation}
\label{eq-M}
	 M \quad \eqdef \quad \Ex \ketbra{\bpsi}{\bpsi}^{\otimes 2} \enspace.
\end{equation}
We use the symmetries of~$M$ in order to do so.
\begin{lemma}
\label{clm:lambda}
Let~$M$ be the operator defined in Eq.~\eqref{eq-M}. Then
\[
M \quad = \quad \beta \left [ \id + \Swap^{AC} \otimes \Swap^{BD} \right ] + \gamma \left [ \id^{AC} \tensor \Swap^{BD} + \Swap^{AC} \tensor \id^{BD} \right ] \enspace,
\]
where~$\beta \eqdef d^{-2} (d^2 - 1)^{-1}$ and $\gamma \eqdef -d^{-3} (d^2 - 1)^{-1}$.
\end{lemma}

\begin{proof}
Since
\[
M \quad = \quad \Ex ~ (\bU^A \tensor \id^B) \ketbra{\mes_d}{\mes_d} (\bU^{\adjoint A} \tensor \id^B) \otimes (\bU^C \tensor \id^D) \ketbra{\mes_d}{\mes_d} (\bU^{\adjoint C} \tensor \id^D) \enspace,
\]
and~$\bU$ is Haar random, the operator~$M$ commutes with~$V^A \tensor V^C \tensor \id^{BD}$ for all~$V \in \unitary( \complex^d)$. Further, since~$( \id \tensor V) \ket{\mes_d} = ( V^\transpose \tensor \id) \ket{\mes_d}$, the operator~$M$ also commutes with~$\id^{AC} \tensor V^B \tensor V^D$. By Lemma~\ref{lem-sum-perm}, we have
\begin{equation}
\label{eq:lincomb}
	M \quad = \quad \alpha \id + \beta (\Swap^{AC} \otimes \Swap^{BD}) + \gamma (\id^{AC} \tensor \Swap^{BD}) + \delta (\Swap^{AC} \tensor \id^{BD}) \enspace. 
\end{equation}
Consider the following linear functionals:
\begin{enumerate}
	\item $X \mapsto \Tr(X)$
	\item $X \mapsto \Tr((\Swap^{AC} \otimes \Swap^{BD}) X)$
	\item $X \mapsto \Tr((\Swap^{AC} \tensor \id^{BD}) X)$
	\item $X \mapsto \Tr((\id^{AC} \tensor \Swap^{BD}) X)$
\end{enumerate}
We apply these functionals to both sides of Eq.~\eqref{eq:lincomb}. We calculate the value of the functional on the left hand side directly from the definition of~$M$, i.e., Eq.~\eqref{eq-M}, and on the right hand side from Eq.~\eqref{eq:lincomb}. We thus obtain the following linear equations:
\begin{enumerate}
	\item $1 = \alpha d^4 + \beta d^2 + \gamma d^3 + \delta d^3$,
	\item $1 = \alpha d^2 + \beta d^4 + \gamma d^3 + \delta d^3$,
	\item $1/d = \alpha d^3 + \beta d^3 + \gamma d^2 + \delta d^4$, and
	\item $1/d = \alpha d^3 + \beta d^3 + \gamma d^4 + \delta d^2$, respectively.
\end{enumerate}
Solving for~$\alpha, \beta, \gamma, \delta$, we get the unique solution
\[
	\alpha = \beta = \frac{1}{d^2 (d^2 - 1)} \enspace, \quad \text{and} \quad \gamma = \delta = \frac{-1}{d^3(d^2 - 1)} \enspace.
\]
\end{proof}

Let~$n \eqdef d^2$. Consider the random vector~$\ket{\bxi_n} \in \complex^d \tensor \complex^d$ defined as~$\ket{\bxi_n} \eqdef d \ket{\bpsi} = d (\bU \tensor \id) \ket{\mes_d}$. We prove that the sequence of these vectors is pseudo-isotropic.

\suppress{
\begin{claim}
\label{clm:variance}
	Let $A$ denote an operator acting on $\C^d \otimes \C^d$ such that $\| A \| \leq 1$. Then
	\[
		\Abs{ \Tr( M A \otimes A) - \frac{1}{d^4} \Tr(A)^2 } \leq \frac{2}{d - 1}.
	\]
\end{claim}
\begin{proof}
	Using Claim~\ref{clm:lambda}, we get that
	\[
	 \Tr( M A \otimes A) =  \beta (\Tr(A)^2 + \Tr(A^2)) + \gamma (\Tr(F^{BD} A \otimes A) + \Tr(F^{AC} A \otimes A)).
	\]
	Using that $\Tr(A)^2 \leq d^4$, $\Tr(A^2) \leq d^2$, and both $\Tr(F^{BD} A \otimes A) \leq d^4$ and $\Tr(F^{AC} A \otimes A)) \leq d^4$ (the latter two follow because $\| F A \otimes A \| \leq 1$), we get
	\begin{align}
		\Abs{ \Tr( M A \otimes A) - \frac{1}{d^4} \Tr(A)^2 } \leq \Abs{ \beta - \frac{1}{d^4} } d^4 + \beta d^2 + 2 |\gamma| d^4  \leq  \frac{2}{d - 1}.
	\end{align}
\end{proof}
}

\begin{lemma}
\label{lem-pseudo-isotropy}
The sequence of vectors~$( \ket{ \bxi_n} )$ is pseudo-isotropic.
\end{lemma}

\begin{proof}
Let~$(A_n \in \linear( \complex^d \tensor \complex^d) : n \ge 1)$ be a sequence of complex matrices with spectral norm~$\norm{A_n}$ bounded by a constant~$\kappa$, for each~$n$. We use the Chebyshev Inequality to show that
\begin{align}
    \label{eq-xin-pseudo-isotropic}
    \frac{1}{n} \Paren{ \bra{\bxi_n} A_n \ket{\bxi_n} - \Tr(A_n) }  \overset{\rP}{\longrightarrow} 0
\end{align}
as~$n \to \infty$. Let~$\bx_n$ be the complex random variable defined as~$\bx_n \eqdef \bra{\bxi_n} A_n \ket{\bxi_n}$. We may verify that~$\expct \density{\bxi_n} = \id$, so that~$\expct \bx_n = \expct \trace( \density{\bxi_n} A_n)  = \trace(A_n)$. Eq.~\eqref{eq-xin-pseudo-isotropic} is equivalent to showing that for every~$\epsilon > 0$, $\Pr( \Abs{\bx_n - \expct \bx_n} > \epsilon n) \to 0$ as~$n \to \infty$. By the Chebyshev Inequality, 
\[
\Pr( \Abs{\bx_n - \expct \bx_n} > \epsilon n) \quad \le \quad \frac{1}{ \epsilon^2 n^2} \expct \Abs{\bx_n - \expct \bx_n}^2 \enspace. 
\]
So it suffices to show that the variance of~$\bx_n$ is~$\order(n^2)$.

The variance~$\expct \Abs{\bx_n - \expct \bx_n}^2 = \expct \Abs{\bx_n}^2 - \Abs{ \expct \bx_n}^2 = \expct \Abs{\bx_n}^2 - \Abs{ \trace( A_n)}^2$. To calculate the second moment of~$\bx_n$, we rewrite it as follows.
\begin{align*}
    \expct \Abs{ \bx_n}^2 \quad 
        & = \quad \expct ~ \bra{\bxi_n} A_n \ket{\bxi_n} \bra{\bxi_n} A_n^\adjoint \ket{\bxi_n} \\
        & = \quad \expct ~ \trace \big[ (\density{\bxi_n} \tensor \density{\bxi_n}) 
            (A_n \tensor A_n^\adjoint) \big] \\
        & = \quad n^2 \trace \big[ M (A_n \tensor A_n^\adjoint) \big] \enspace,
\end{align*}
where~$M$ is the matrix defined in Eq.~\eqref{eq-M}. By Lemma~\ref{clm:lambda}, and the H{\"o}lder Inequality (namely, $\Abs{ \trace(AB)} \le \trnorm{A} \norm{B}$),
\begin{align*}
    \expct \Abs{ \bx_n}^2 \quad 
        & = \quad n^2 \Big[ \beta \big( \Abs{\trace( A_n)}^2 + \trace( A_n^\adjoint A_n) \big) \\
        & \qquad \mbox{} + \gamma \big( \trace \big( (  \Swap^{AC} \tensor \id^{BD}) 
            ( A_n \tensor A_n^\adjoint) \big) + \trace \big( ( \id^{AC} \tensor \Swap^{BD}) ( A_n \tensor A_n^\adjoint)  \big) \big) \Big] \\
        & \le \quad n^2 \Big[ \beta \big( \Abs{\trace( A_n)}^2 + \kappa^2 n \big) + 2 \Abs{\gamma} \kappa^2 n^2 \Big] \enspace,
\end{align*}    
where~$\beta = 1/n(n - 1)$ and~$\gamma = -1/n^{3/2} (n - 1)$.
Thus the variance is bounded as
\begin{align*}
    \expct \Abs{ \bx_n}^2 - \Abs{ \trace( A_n)}^2 \quad 
        & \le \quad \frac{1}{ n - 1} \Abs{\trace( A_n)}^2 + \frac{ \kappa^2 n^2}{ n - 1} + \frac{ 2 \kappa^2 n^{5/2} }{ n - 1} \enspace,
\end{align*}
which is~$\order( n^2)$ as~$ \Abs{\trace( A_n)} \le \kappa n$. This proves that the sequence~$(\ket{\bxi_n})$ is pseudo-isotropic.
\end{proof}

\subsection{Analysis of a random protocol}
\label{sec-random-protocol-analysis}

Consider the following random protocol~$\bPi_d$. Let~$d$ be an integer~$\ge 2$, and~$n \eqdef d^2$.
Alice and Bob agree on a choice of~$n$ independently chosen Haar-random unitary operators~$\bU_1, \dotsc, \bU_n \in \unitary(\complex^d)$. They also share the maximally entangled state $\ket{\mes_d} \coloneqq \tfrac{1}{\sqrt{d}} \sum_{k = 1}^d \ket{k}\ket{k}$ with local dimension~$d$.
When Alice gets message $i \in [n]$, she applies $\bU_i$ to her half of $\ket{\mes_d}$, and sends it over to Bob. Bob now holds the state~$\ket{\bpsi_i} \eqdef (\bU_i \otimes \id ) \ket{\mes_d}$. He performs an optimal measurement to identify~$i$, given that the state is drawn from the ensemble~$\bcE_d \eqdef \big( \ket{\bpsi_j} : j \in [n] \big)$.

Aram Harrow (personal communication) suggested the protocol~$\bPi_d$ as a candidate for an approximate~$(d, \epsilon)$-superdense coding protocol with vanishing error~$\epsilon$ in the limit of large dimension. If this random construction of superdense coding protocols did indeed have error that vanishes rapidly as a function of~$d$, then this could potentially refute \Cref{conj:robust-rigidity}. This is formalized by the following proposition (which was stated in \Cref{sec-robust-rigidity}, and is reproduced here for convenience).

\counterexample*

\begin{proof}
Let $\bPi_d$ be the random protocol and let $(\bU_i)$ be the ensemble of random unitaries specified in the Proposition statement. Suppose for contradiction that \Cref{conj:robust-rigidity} were true and the error~$\beps$ of the protocol $\bPi_d$ satisfied
\begin{equation}
\label{eq:counterexample-0}
	\Ex_{(\bU_i)} \delta_2(\beps)^2 \quad < \quad (2d)^{-2}~.
\end{equation}
First we argue that
\begin{equation}
\label{eq:counterexample-1}
\Ex_{(\bU_i)} \Ex_{\bj \neq \bk}  | \Tr(\bU_\bj \bU_\bk^*) |^2 \quad = \quad 1 \enspace,
\end{equation}
where the first expectation is over the ensemble of random unitary operators $(\bU_i)$, and the second expectation is over a uniformly random pair of distinct indices~$\bj,\bk \in [d]$, $\bj \neq \bk$. To prove this, note that for all~$j \neq k$ $\Ex_{(\bU_i)} |\Tr(\bU_j \bU_k^*)|^2 = \Ex_{(\bU_i)} |\Tr(\bU_1 \bU_2^*)|^2$ because~$\bU_i$ are independent, identically distributed Haar-random unitaries operators. Furthermore, by the rotation invariance of the Haar measure, $\bU_1 \bU_2^*$ is also distributed according to the Haar measure. So the above quantity is equal to $\Ex_\bU |\Tr( \bU)|^2$ for Haar-random $\bU$. So the LHS of \Cref{eq:counterexample-1} equals
\begin{align*}
	\Ex_\bU |\Tr(\bU)|^2 \quad & = \quad \Ex_\bU \Abs{ \sum_{j = 1}^d \bra{j} \bU \ket{j} }^2 \\
				   &= \quad \Ex_\bU \sum_{j,k=1}^d \bra{j} \bU \ketbra{j}{k} \bU^\adjoint \ket{k} \\
				   &= \quad \sum_{j,k} \bra{j} \Big( \Ex_\bU \bU \ketbra{j}{k} \bU^\adjoint \Big) \ket{k}~.
\end{align*}
Since $\Ex_\bU \bU \ketbra{j}{j} \bU^* = \id/d$ and $\Ex_\bU \bU \ketbra{j}{k} \bU^* = 0$ when $j \neq k$, we have $\Ex |\Tr(\bU)|^2 = 1$, which establishes \Cref{eq:counterexample-1}.

On the other hand, the rigidity condition promised by \Cref{conj:robust-rigidity} implies that every collection of $d \times d$ unitary operators $(U_i)$ yields a superdense protocol with some error $\eps$, and in turn there exists an orthogonal unitary basis $(E_i)$ such that $\| U_i - E_i \|_{\nhs} \leq \delta_2(\eps)$ for all $i \in [d^2]$. Note that $\eps$ and $(E_i)$ depend on $(U_i)$, and let~$\beps$ and~$(\bE_i)$ be the error of the protocol~$\bPi_d$ and the corresponding orthogonal unitary basis, respectively. Then
\begin{align}
	\Ex_{(\bU_i)} \Ex_{\bj \neq \bk} |\Tr(\bU_\bj \bU_\bk^*)|^2 \quad &= \quad \Ex_{(\bU_i)} \Ex_{\bj \neq \bk} \left| \Tr(\bU_\bj \bU_\bk^*) - \Tr(\bE_\bj \bE_\bk^*) \right|^2 \notag \\
	    &\leq \quad \Ex_{(\bU_i)} \Ex_{\bj \neq \bk} \Big( \left| \Tr((\bU_\bj - \bE_\bj) \bU_\bk^*)| + |\Tr(\bE_\bj (\bU_\bk^* - \bE_\bk^*)) \right| \Big)^2 \notag \\
		&\leq \quad 2 \Ex_{(\bU_i)} \Ex_{\bj \neq \bk} \left|\Tr((\bU_\bj - \bE_\bj) \bU_\bk^*) \right|^2 + \left| \Tr(\bE_\bj (\bU_\bk^* - \bE_\bk^*)) \right|^2
			\label{eq:counterexample-2}
\intertext{where the first equality is due to the orthogonality condition~$\Tr(\bE_j \bE_k^*) = 0$ whenever~$j \neq k$, the second line is due to the triangle inequality, and the third line is due to the inequality~$(a + b)^2 \leq 2a^2 + 2b^2$ for real numbers $a,b$. By the Cauchy-Schwarz inequality for the Hilbert-Schmidt inner product, we have $|\Tr((\bU_j - \bE_j) \bU_k^*)|^2 \leq \| \bU_j - \bE_j \|_2^2 \cdot \| \bU_k \|_2^2$ and $|\Tr(\bE_j (\bU_k^* - \bE_k^*))|^2 \leq \| \bE_j \|_2^2 \cdot \| \bU_k - \bE_k \|_2^2$, where $\| X \|_2 = \sqrt{\Tr(XX^\adjoint)}$ denotes the (unnormalized) Hilbert-Schmidt norm. Since $\| A \|_2^2 = d$ for all $d \times d$ unitary matrices $A$, we can upper bound the RHS of \Cref{eq:counterexample-2} as}
    &\leq \quad 2 d \Ex_{(\bU_i)} \Ex_{\bj \neq \bk} \left( \| \bU_\bj - \bE_\bj \|_2^2 + \| \bU_\bk - \bE_\bk \|_2^2 \right) \notag \\
\suppress{
    &= \quad 4 d \Ex_{(\bU_i)} \Ex_{j} \| \bU_j - \bE_j \|_2^2 \notag \\
}
    &= \quad 4 \Ex_{(\bU_i)} \sum_{j} \| \bU_j - \bE_j \|_{\nhs}^2 \notag \\
    &\leq \quad 4d^2 \Ex_{(\bU_i)} \delta_2(\beps)^2 \enspace \notag \\
    & < \quad 1 \enspace, \notag
\end{align}
where the last inequality follows from the assumption in \Cref{eq:counterexample-0}. However, this contradicts \Cref{eq:counterexample-1}. 
Thus, either the conjecture does not hold, or the random superdense coding protocol $\bPi_d$ has error satisfying $\Ex \delta_2(\beps) \geq (2d)^{-2}$. 
\end{proof}

In the rest of this section, we prove that for sufficiently large dimension, with high probability, the protocol~$\bPi_d$ has positive constant error. This indicates that random maximally entangled quantum states are not very reliable for transmitting classical information, and proves Theorem~\ref{thm:random-protocol}. Thus, the random protocol~$\bPi_d$ does not rule out a robust rigidity theorem for superdense coding.

To analyze the decoding error of the protocol, we study the \emph{distinguishability\/} of the ensemble~$\bcE_d$. This is the probability that, if the pure state~$\ket{\bpsi_i}$ is selected uniformly at random from the ensemble, an optimal measurement correctly identifies the state.
\begin{definition}
Let~$\cF \eqdef \big( ( p_i, \rho_i ) : \rho_i \in \qstate( \complex^k), i \in [m] \big) $ be an ensemble of states in which state~$\rho_i$ occurs with probability~$p_i$. 
We define the \emph{distinguishability\/} of~$\cF$ as
\[
	\dist(\cF) \quad \eqdef \quad \max_{\text{POVM } M} \sum_{i = 1}^m p_i \Tr(M_i \rho_i) \enspace,
\]
where the maximization is over all measurements (i.e., POVMs)~$M$ with elements~$M_1, \ldots, M_m$.
\end{definition}
We can estimate the distinguishability of an ensemble of states via the generalized Holevo-Curlander bounds~\cite{Holevo79-distinguishability,Curlander79-distinguishability,ON99-converse-channel-coding,Tyson09-distinguishability}.
\begin{theorem}[generalized Holevo-Curlander bounds~\cite{ON99-converse-channel-coding,Tyson09-distinguishability}]
\label{thm:holevo-curlander}
Let~$\cF \eqdef \big( ( p_i, \rho_i ) : \rho_i \in \qstate( \complex^k), i \in [m] \big)$ be an ensemble of~$m$ quantum states. Then the distinguishability of~$\cF$ satisfies
\[
	\left( \hc(\cF) \right)^2 \quad \leq \quad \dist(\cF) \quad \leq \quad \hc(\cF) \enspace,
\]
where
\[
	\hc(\cF) \quad \eqdef \quad \Tr \sqrt{ \sum_{i = 1}^m p_i^2 \rho_i^2 } \enspace.
\]
\end{theorem}
We only need the upper bound on distinguishability above for a uniform ensemble of pure states. This bound was given by Curlander~\cite{Curlander79-distinguishability} in the case of linearly independent states. It was generalized to the case of equiprobable, possibly mixed states by Ogawa and Nagaoka~\cite[Lemma~1]{ON99-converse-channel-coding}. The proof they gave also extends with minor modifications to non-uniform ensembles. The two bounds in \Cref{thm:holevo-curlander} were proven --- re-proven independently in the case of the upper bound~\cite{Tyson09-erratum} --- by Tyson\cite[Theorem 10]{Tyson09-distinguishability}. Tyson later gave another proof of the bounds which also generalizes to error-recovery~\cite[Section III]{Tyson10-error-recovery}.

We show that the expectation of the quantity~$\hc(\bcE_d)$ for the ensemble of random maximally entangled states is at most a constant strictly less than~$1$, for sufficiently large dimension~$d$. This implies that the distinguishability~$\dist(\bcE_d)$ is also strictly less than~$1$ in expectation, and that any measurement Bob makes has a non-zero constant probability of failure, on average. 

\begin{theorem}
\label{thm-protocol-error}
The distinguishability~$\dist( \bcE_d)$ of the random superdense coding protocol~$\bPi_d$ tends to~$ \tfrac{8}{3\pi}\approx 0.85 $ as~$d \to \infty$.
\end{theorem}
\begin{proof}
Define the matrix~$\bQ$ as~$\bQ \coloneqq \sum_{i = 1}^n \density{\bpsi_i}$, and~$\bL_d$ as a uniformly random eigenvalue of~$\bQ$. Then~$\expct \hc( \bcE_d)$ is the expectation of the random variable $\sqrt{\bL_d}$, and we aim to bound this from above.

Define the~$n \times n$ matrix~$\bR \eqdef d \sum_{i=1}^n \ketbra{\bpsi_i}{i}$ so that~$\bQ = \tfrac{1}{n} \bR \bR^\adjoint$. Consider the random vector~$\ket{\bxi_n}$ defined as~$\ket{\bxi_n} \eqdef d (\bU \tensor \id) \ket{\mes_d}$, where~$\bU \in \unitary(\complex^d)$ is a Haar-random unitary operator. I.e., $\ket{\bxi_n}$ is a scaled random maximally entangled state. Lemma~\ref{lem-pseudo-isotropy} shows that the sequence~$(\ket{\bxi_n})$ is pseudo-isotropic. Since the columns of the matrix~$\bR$ are i.i.d.\ copies of the random vector~$\ket{\bxi_n}$, the random matrix~$\bQ$ is of the form described in Theorem~\ref{thm:marcenko-pastur}. Thus the limiting distribution of the uniformly random eigenvalue~$\bL_d$ of~$\bQ$ follows the Mar\v{c}enko-Pastur law with density~$\MP_1$ (i.e., with parameter~$r = 1$):
\begin{align}
\label{eq-mp1}
\MP_1(x) \quad & \coloneqq \quad 
\begin{cases}
    \frac{1}{2 \pi} \sqrt{ \frac{ 4 - x} {x} }
        & \text{ if } 0 \le x \le 4 \enspace, \text{ and} \\
    0 & \text{otherwise} \enspace.
\end{cases}
\end{align}
We would like to use the density~$\MP_1$ to estimate the limit of the expectation of~$\sqrt{\bL_d}$.
A subtle issue is that weak convergence (i.e., convergence in distribution of the random variables) does not necessarily
imply that the limit of the expectation values~$\expct \sqrt{ \bL_d}$ equals the expectation of the
limiting random variable. A simple example for which this does not hold is described in Section~\ref{sec-expectation-limit}. Nonetheless, in Theorem~\ref{thm-expectation-limit} in Section~\ref{sec-expectation-limit}, we show that the sequence of random variables~$\bL_d$ satisfies the stronger property we need. Namely, $\expct \sqrt{ \bL_d}$ converges to~$\expct \sqrt{ \bL}$ as~$d \to \infty$, where~$\bL$ is a random variable with density~$\MP_1$. We may thus bound the distinguishability of the ensemble~$\bcE_d$ as~$d \to \infty$ as follows.
\begin{align*}
\lim_{d \to \infty} \expct \dist( \bcE_d) \quad & \le \quad \lim_{d \to \infty} \expct \hc( \bcE_d) \\
    & = \quad \lim_{d \to \infty} \expct \sqrt{ \bL_d} \quad = \quad \expct \sqrt{ \bL} \\
    & = \quad \int_{-\infty}^{\infty} \sqrt{x} \, \MP_1(x) \diff x \\
    & = \quad \frac{1}{2 \pi} \int_0^4 \sqrt{4 - x} \diff x \\
	& = \quad \frac{8}{3 \pi} \quad \approx \quad 0.85 \enspace.
\end{align*}
\suppress{
We verify that the columns of $R$ satisfy the conditions of the Mar\v{c}enko-Pastur law. The columns of the matrix $R$ are drawn from $\cal{D}$, scaled by a factor of $d$. Let $\hat{\cal{D}}$ denote this scaled distribution. Then Section~\ref{sec:random_mes} shows
\[
	\Ex_{\ket{v} \sim \hat{\cal{D}}} \ketbra{v}{v} = \Ex_{\ket{\psi} \sim \cal{D}} d^2 \ketbra{\psi}{\psi} = \id.
\]
Next, we use Markov's inequality to bound the probability
\begin{align}
\Pr\left ( \frac{1}{d^4} \Paren{ \Tr(A \ketbra{v}{v}) - \Tr(A)}^2 \geq \eps \right) &\leq
\frac{ \frac{1}{d^4}  \Paren{\Ex \Tr(A \ketbra{v}{v})^2 - \Paren{\Ex \Tr(A \ketbra{v}{v}) }^2}}{\eps} \\
&=\frac{ \Paren{\Ex \Tr(A \ketbra{\psi}{\psi})^2 - \frac{1}{d^4} \Paren{\Ex \Tr(A) }^2}}{\eps} \\
&\leq \frac{2}{\eps} (d-1)^{-1}
\end{align}
where the last line follows from Claim~\ref{clm:variance}. As $d \to \infty$, this probability goes to $0$. Thus we can apply the Mar\v{c}enko-Pastur law. We have that
\[
	d^{-3} \Ex \Tr \left (\sqrt{RR^\adjoint} \right) = d^{-2} \Ex \Tr \left (\sqrt{\frac{1}{d^2} RR^\adjoint} \right) = \Ex \sum_i \frac{1}{d^2} \sqrt{ \lambda_i \left (\frac{1}{d^2} RR^\adjoint \right )} \to \int_{-\infty}^{\infty} \sqrt{x} \, p(x) \rd x
\]
where $p(x)$ is the density given by Theorem~\ref{thm:marcenko-pastur}.
}
\end{proof}

We can strengthen this result to show that the distinguishability of~$\bcE_d$ is tightly concentrated around the mean. So all but an exponentially small fraction of the superdense coding protocols using~$\ket{\mes_d}$ succeed with probability smaller than~$\approx 0.85$.
\begin{theorem}
\label{thm-error-concentration}
The distinguishability~$\dist( \bcE_d)$ of the random superdense coding protocol~$\bPi_d$ satisfies
\[
 \Pr \Big ( \dist( \bcE_d) \geq \Ex \dist( \bcE_d) + t \Big) \leq \exp \! \left( - \frac{d^3 (d-2)t^2}{96} \right) \enspace.
\]
\end{theorem}
\begin{proof}
Define the function~$f :  \big( \unitary( \complex^d) \big)^n \rightarrow \reals$ as
\[
    f(U_1, \ldots, U_n) \quad \eqdef \quad \sup_{\text{POVM } M} \frac{1}{n} \sum_{i = 1}^n \Tr(M_i \psi_i) \enspace,
\]
where $\ket{\psi_i} = (U_i \otimes \id) \ket{\mes_d}$, and we denote~$\density{ \psi}$ by~$\psi$. So~$f(\bU_1, \ldots, \bU_n) = \dist( \bcE_d) $, the distinguishability of the ensemble $\bcE_d$. We bounded the expected distinguishability by~$ \approx .85 < 1$ in Theorem~\ref{thm-protocol-error}. We  show that~$\dist( \bcE_d)$ is tightly concentrated around its expectation using Theorem~\ref{thm-Haar-concentration}. To do so, we compute a bound on the Lipschitz constant of~$f$. 

Fix unitary operators $U_1,\ldots,U_n, U_1',\ldots,U_n' \in \unitary( \complex^d ) $ and let $\ket{\psi_i} = (U_i \otimes \id) \ket{\mes_d}$ and $\ket{\psi_i'} = (U_i' \otimes \id) \ket{\mes_d}$. Since the space of $n$-dimensional POVMs with~$n$ outcomes is compact, for any sequence~$(U_i)$ the supremum in the definition of~$f$ is attained at some POVM~$M$. Let $M$, $M'$ correspond to the POVMs achieving~$f(U_1,\ldots,U_n)$ and~$f(U_1',\ldots,U_n')$, respectively, and let~$\alpha, \alpha'$ denote these quantities. Assume without loss of generality that $\alpha' \leq \alpha$.

We have that
\begin{align*}
|\alpha - \alpha'| \quad = \quad \alpha - \alpha' \quad = \quad \frac{1}{d^2} \Big( \sum_i \Tr(M_i \psi_i) - \Tr(M_i' \psi_i') \Big) \quad \leq \quad \frac{1}{d^2} \sum_i \Tr(M_i (\psi_i - \psi_i')) \enspace,
\end{align*}
as the POVM~$M$ may not be an optimal distinguishing measurement for the ensemble~$(\psi_i')$. We bound this by
\begin{align*}
\MoveEqLeft \frac{1}{d^2} \sum_i \Big | \Tr(M_i ( \psi_i - \psi_i') ) \Big | \\
    &\leq \quad \frac{1}{d^2} \sum_i \| M_i \| \cdot \| \psi_i - \psi_i' \|_1 & (\text{H\"{o}lder inequality})\\
    &\leq \quad \frac{1}{d^2} \sum_i \| \psi_i - \psi_i' \|_1 & (\text{$M$ is a POVM}) \\
    &\leq \quad \frac{1}{d^2} \sum_i 2 \| \ket{\psi_i} - \ket{\psi_i'} \| \\
    &=  \quad \frac{2}{d^2} \sum_i  \sqrt{ \bra{\mes} ( U_i - U_i')^\adjoint (U_i - U_i') \ket{\mes} } \\
    &\leq \quad 2 \sqrt{ \frac{1}{d^2} \sum_i \frac{1}{d} \| U_i - U_i' \|_2^2 } & (\text{Jensen inequality})
\end{align*}
where in the fourth line we used the property that for any two pure states $\ket{\varphi}$ and $\ket{\theta}$, the trace distance between $\ketbra{\varphi}{\varphi}$ and $\ketbra{\theta}{\theta}$ is at most $2 \| \ket{\varphi} - \ket{\theta} \|$. In the last line, we used the identity~$\bra{\mes_d} (A \otimes \id) \ket{\mes_d}$ for any~$d \times d$ matrix is equal to $\Tr(A)/d$. 

Thus $|\alpha - \alpha'| \leq 2 d^{-3/2} \sqrt{ \sum_i \| U_i - U_i' \|_2^2 }$, which implies that~$f$ is~$2 d^{-3/2}$-Lipschitz. Applying Theorem~\ref{thm-Haar-concentration} we obtain
\[
 \Pr \Big ( \dist( \bcE_d) \geq \Ex \dist( \bcE_d) + t \Big) \quad \leq \quad \exp \! \left ( - \frac{d^3 (d-2)t^2}{96} \right) \enspace.
\]
\end{proof}

\subsection{A subtle issue}
\label{sec-expectation-limit}

Recall from the proof of Theorem~\ref{thm-protocol-error} in Section~\ref{sec-random-protocol-analysis} that~$\bL_d$ denotes a uniformly random eigenvalue of the matrix~$\bQ$. The generalized Mar{\v c}enko-Pastur Law (Theorem~\ref{thm:marcenko-pastur}) tells us that~$\bL_d$ converges in distribution to a random variable~$\bL$ with density~$\MP_1$ given in Eq.~\eqref{eq-mp1} as~$d \rightarrow \infty$. We used this limiting distribution to estimate the limit of the mean of~$\sqrt{\bL_d}$ in Theorem~\ref{thm-protocol-error}. We pointed out the subtle issue that convergence in distribution does not necessarily imply that the limit of means equals the mean of the limiting random variable. A simple example which illustrates this issue is the following. For any positive integer~$k$, let the random variable~$\bx_k$ take value~$k$ with probability~$1/k$, and value~$0$ with the remaining probability. Then~$\bx_k$ converges in distribution to the constant~$0$, whereas~$\expct \bx_k = 1$ for all~$k$.

The example above highlights the reason underlying this phenomenon: while the probability of an interval on the line may go to zero in the limit, the rate of convergence may not be fast enough to dampen the contribution to the mean from that interval. We show that the probability that the random variable~$\bL_d$ deviates from zero decays exponentially. This helps us conclude the convergence of the mean~$\expct \sqrt{\bL_d}$ to~$\expct \sqrt{\bL}$. 

A similar property was assumed to hold by Montanaro~\cite{montanaro2007distinguishability} in his work on the distinguishability of random quantum states. Let~$\bS \eqdef \sum_{i = 1}^k \ketbra{\bzeta_i}{i}$, where~$\ket{\bzeta_i}$ are i.i.d.\ random vectors in~$\complex^d$ with i.i.d.\ complex gaussian entries with mean~$0$ and variance~$1$. Montanaro approximates~$\expct \trace( \bS \bS^\adjoint)^{1/2}$ using the Mar{\v c}enko-Pastur Law. He justifies this using estimates on the rate of convergence of the expected distribution of a uniformly random eigenvalue of~$(1/k) \bS \bS^\adjoint$ to the limiting distribution given by the Mar{\v c}enko-Pastur Law (see the discussion after Lemma~5 in Ref.~\cite{montanaro2007distinguishability}). The rate of convergence is measured in terms of the Kolmogorov distance between the two distributions. (The Kolmogorov distance between the cumulative distribution functions~$F_1$ and~$F_2$ of real random variables is defined as~$\sup_{x \in \reals} \Abs{F_1(x) - F_2(x)}$.) The Kolmogorov distance was shown to be~$\Order( k^{-5/48})$ by Bai~\cite{Bai93-convergence-rate}. However, vanishing Kolmogorov distance does not necessarily imply the convergence of the mean to the mean of the limiting distribution. For example, the Kolmogorov distance of the distribution of the random variable~$\bx_k$ defined above from the constant~$0$ is~$1/k$. The approach we take in this section can also be used to fill the gap in Montanaro's work. In fact, the analogue of Lemma~\ref{lem-norm-tail} we need for this purpose follows directly, as the columns of~$\bS$ are gaussian. We leave the details to the reader, and return to the analysis of the random variable~$\bL_d$.

In order to show that~$\bL_d$ has an exponentially decaying tail, it suffices to show that the spectral norm of the matrix~$\bQ$---i.e., its largest eigenvalue---has this property. So we proceed by deriving a tail bound for the spectral norm of~$\bQ$.
\begin{lemma}
\label{lem-norm-tail}
Let~$d \ge 3$. There are positive universal constants~$c_1, c_2$ such that for all $t \geq c_1$, 
\begin{equation}
\label{eq-norm-bound}
\Pr \left ( \norm{ \bQ } > t \right ) \quad \le \quad 2 \exp( - tn / c_2) \enspace.
\end{equation}
\end{lemma}
\begin{proof}
Recall that~$n \coloneqq d^2$, $\bQ \coloneqq \sum_{ i = 1 }^n \density{ \bpsi_i }$, and that~$\bR \eqdef d \sum_{i = 1}^n \ketbra{\bpsi_i}{i}$. We have~$\norm{\bQ} = n^{-1} \norm{ \bR \bR^\adjoint} = n^{-1} \norm{ \bR}^2$, so
\[
\Pr \left( \norm{ \bQ } > t \right) \quad = \quad \Pr \left( \norm{ \bR } > \sqrt{nt} \right) \enspace.
\]
It thus suffices to give a suitable tail bound for~$\norm{ \bR}$.

The vectors~$d \ket{\bpsi_i}$ are i.i.d.\ copies of the random vector~$\ket{\bxi_n}$ defined as~$\ket{\bxi_n} \eqdef d ( \bU \tensor \id ) \ket{\mes_d}$, where~$\bU$ is a Haar-random unitary operator on~$\complex^d$. So the vectors~$d \ket{\bpsi_i}$ have zero mean. We can verify that~$\expct \density{ \bxi_n } = \id$, so the vectors~$d \ket{\bpsi_i}$ are isotropic. We prove below that the vector~$\ket{\bxi_n}$ is sub-gaussian, with sub-gaussian norm at most a universal constant~$\kappa$. So the matrix~$\bR^\adjoint = d\sum_{i=1}^n \ketbra{i}{\psi_i}$ satisfies the conditions of Theorem~\ref{thm-sub-gaussian}, and we have that for some positive universal constant~$c_3$,
\[
	\norm{\bR} = \norm{ \bR^\adjoint} \quad > \quad \sqrt{n} + c_3 \kappa^2 ( \sqrt{n}  + t_1 )
\]
with probability at most~$2 \exp( - t_1^2 )$, for all~$t_1 \geq 0$. Let~$t_1$ be such that the right hand side above equals~$\sqrt{nt}$, i.e.,
\[
t_1 \quad \eqdef \quad \frac{ \sqrt{n} }{ c_3 \kappa^2 } \left( \sqrt{t} - (1 + c_3 \kappa^2) \right) \enspace.
\]
Let~$ c_1 \eqdef 4 (1 + c_3 \kappa^2)^2 $. Whenever~$ t \ge c_1 $, we see that~$ t_1 \ge \sqrt{ nt} / 2 c_3 \kappa^2 \ge 0 $. So
\[
\Pr \left( \norm{ \bR } > \sqrt{nt} \right) \quad \le \quad 2 \exp \! \big( - nt / 4 c_3^2 \kappa^4 \big) \enspace,
\]
and the theorem holds with~$ c_2 \eqdef 4 c_3^2 \kappa^4$.

It remains to prove that the sub-gaussian norm~$\kappa$ of~$\ket{\bxi_n}$ is at most some universal constant. By Lemma~\ref{lem-equiv-sub-gaussian}, it suffices to show that for any unit vector~$\ket{u} \in \complex^n$, the random variable~$\bx \eqdef \braket{u}{\bxi_n}$ has sub-gaussian tails: for a positive universal constant~$\kappa_1$,
\begin{equation}
\label{eq-sg-tail}  
\Pr( \Abs{ \bx } \ge t ) \quad \le \quad 2 \exp \! \big( - t^2 / \kappa_1^2 \big) \enspace,
\end{equation}
for all~$t \ge 0$. We establish this by appealing to Theorem~\ref{thm-Haar-concentration}.

Since~$\ket{ \bxi_n}$ is isotropic,
\[
    \expct \Abs{ \bx}^2 \quad = \quad \bra{u} \; (\expct \density{ \bxi_n} ) \; \ket{u} 
        \quad = \quad 1 \enspace.  
\]
So~$\expct \Abs{ \bx} \le \big( \expct \Abs{ \bx}^2 \big)^{1/2} \le 1 $.

Define the function~$f : \unitary( \complex^d) \rightarrow \complex$ as~$f(U) \eqdef d \Abs{ \bra{u} ( U \tensor \id ) \ket{\mes_d} }$. Then~$\Abs{ \bx} = f( \bU)$. To show that~$f$ is Lipschitz, consider~$U, V \in \unitary( \complex^d)$. Since~$\ket{ u}$ is a unit vector and~$\bra{ \mes_d} ( W \tensor \id ) \ket{ \mes_d} = \tfrac{1}{d} \trace( W)$ we have
\begin{align*}
    \Abs{ f(U) - f(V)} \quad & \le \quad d \Abs{ \bra{u} ( (U - V) \tensor \id ) \ket{ \mes_d} } \\
        & \le \quad d \norm{ ( (U - V) \tensor \id ) \ket{ \mes_d} } \\
        & = \quad d \left( \bra{ \mes_d} ( (U - V)^\adjoint (U - V) \tensor \id ) \ket{ \mes_d} \right)^{1/2} \\
        & = \quad \sqrt{d} \norm{ U - V }_2 \enspace,
\end{align*}
where~$\norm{ \cdot}_2$ denotes the Hilbert-Schmidt norm on~$\linear( \complex^d)$. So the function~$f$ is~$\sqrt{d}$-Lipschitz with respect to Hilbert-Schmidt metric. By Theorem~\ref{thm-Haar-concentration},
\[
\Pr( \Abs{ \bx } \ge 1 + t_1 ) \quad \le \quad \exp \! \left( - \frac{ (d - 2) t_1^2 }{ 24d} \right)
\]
for every~$t_1 > 0$. For~$d \ge 3$, we have~$d - 2 \ge d/3$. So the right hand side is at most~$\exp( - t_1^2 / 72)$, and we have
\[
\Pr( \Abs{ \bx } \ge t ) \quad \le \quad \exp \! \left( - \frac{(t - 1)^2 }{3 \cdot 24 } \right) \quad \le \quad \exp \! \left( - \frac{t^2 }{12 \cdot 24 } \right) \enspace,
\]
for all~$t \ge 2$ (as~$t - 1 \ge t/2$). Note that Eq.~\eqref{eq-sg-tail} holds trivially for~$t \in [0,2]$ for any choice of positive constant~$\kappa_1$ such that~$2 \exp(- 4 / \kappa_1^2) \ge 1$. Thus, taking~$ \kappa_1 \eqdef 24 $, we see that Eq.~\eqref{eq-sg-tail} holds for all~$ t \ge 0$ whenever~$d \ge 3$.
\end{proof}

Lemma~\ref{lem-norm-tail} implies that for a large enough constant~$\alpha$, the contribution to the mean~$\expct \sqrt{ \bL_d}$ outside an interval~$[0, \alpha]$ goes to~$0$ as~$d \to \infty$. Within this interval, the contribution to the mean tends to that for~$\expct \sqrt{ \bL}$. This helps us derive the limiting value of the mean. 
\begin{theorem}
\label{thm-expectation-limit}
$
\lim_{d \rightarrow \infty} \expct \sqrt{ \bL_d}
    \quad = \quad \expct \sqrt{ \bL} \enspace.
$
\end{theorem}
\begin{proof}
We formalize the intuition given above by appealing to a weaker property implied by
convergence in distribution, namely that the expectation of any \emph{bounded\/}
continuous function~$f$ of the random variable~$\bL_d$ converges to~$\expct f(\bL)$. 

Fix~$\alpha \ge \max \set{ c_1, 4}$, where~$c_1$ is the constant in the statement of Lemma~\ref{lem-norm-tail} and consider the function~$f_\alpha$ defined as follows:
\[
f_\alpha(x) \quad \coloneqq  \quad 
    \begin{cases}
        0 & \quad x \le 0 \\
        \sqrt{x} & \quad 0 < x \le \alpha \\
        \sqrt{\alpha} & \quad \alpha < x \enspace.
    \end{cases}
\]
Since~$f_\alpha$ is continuous and bounded, and~$\bL \in [0, 4]$,
\[
\lim_{d \to \infty} \expct f_\alpha ( \bL_d) \quad = \quad \expct f_\alpha ( \bL)
    \quad = \quad \expct \sqrt{ \bL} \enspace.
\]
On the other hand,~$ \bL_d \ge 0 $ and~$f_\alpha(x) \le \sqrt{x} $ for all~$x \ge 0$. So~$\expct f_\alpha( \bL_d) \le \expct \sqrt{ \bL_d} $, and
\[
\expct \sqrt{ \bL} \quad = \quad \lim_{d \to \infty} \expct f_\alpha ( \bL_d)
    \quad \le \quad \lim_{d \to \infty} \expct \sqrt{ \bL_d} \enspace.
\]

We prove the reverse inequality using Lemma~\ref{lem-norm-tail}. Let~$p(x)$ be the probability density function of~$\bL_d$. By the definition of~$f_\alpha$,
\begin{align}
\label{eq-expct-bd}
    \expct \sqrt{ \bL_d} \quad & \le \quad \expct f_\alpha( \bL_d) + \int_{x \ge \alpha} \sqrt{x} p(x) \diff x \enspace.
\end{align}
Let~$g( \alpha, d)$ denote the second term on the right hand side of Eq.~\eqref{eq-expct-bd} above. This is the contribution to~$\expct \sqrt{\bL_d}$ outside of the interval~$[0, \alpha]$. Using~$\alpha \ge 4$, Fubini's Theorem, $ \bL_d \le \norm{ \bQ}$, and Lemma~\ref{lem-norm-tail}, we have
\begin{align*}
    g( \alpha, d) \quad & \le \quad \int_{x \ge \alpha} x p(x) \diff x \\
        & = \quad \int_{x \ge \alpha} \int_{y \in [0,x]} p(x) \diff y \diff x \\
        & = \quad \int_{y \ge 0} \int_{x \ge \max \set{ \alpha, y}} p(x) \diff x \diff y \\
        & = \quad \int_{y \in [0, \alpha]} \int_{x \ge \alpha} p(x) \diff x \diff y 
            + \int_{y \ge \alpha} \int_{x \ge y} p(x) \diff x \diff y \\
        & = \quad \int_{y \in [0, \alpha]} \Pr( \bL_d \ge \alpha) \diff y 
            + \int_{y \ge \alpha} \Pr( \bL_d \ge y) \diff y \\
        & \le \quad 2 \alpha \exp( - \alpha n/ c_2) + 2 \int_{y \ge \alpha} \exp( - y n/ c_2) \diff y \\
        & = \quad 2 ( \alpha + c_2 / n )  \exp( - \alpha n/ c_2) \enspace, 
\end{align*}
where~$c_2$ is the universal constant in the statement of Lemma~\ref{lem-norm-tail}. Since~$n = d^2$, $ g( \alpha, d) $ vanishes as~$d$ goes to~$\infty$. By Eq.~\eqref{eq-expct-bd},
\begin{align*}
    \lim_{d \to \infty} \expct \sqrt{ \bL_d} \quad 
        & \le \quad \lim_{d \to \infty} \expct f_\alpha( \bL_d) +   \lim_{d \to \infty} g( \alpha, d) \\
        & = \quad \expct f_\alpha( \bL) \quad = \quad \expct \sqrt{ \bL} \enspace.
\end{align*}
This proves the theorem.
\end{proof}
 
\bibliography{bibl}

\newcommand{\etalchar}[1]{$^{#1}$}
\begin{thebibliography}{BCWdW01}

\bibitem[Bai93]{Bai93-convergence-rate}
Z.~D. Bai.
\newblock Convergence rate of expected spectral distributions of large random
  matrices. {Part II}. sample covariance matrices.
\newblock {\em Annals of Probability}, 21(2):649--672, April 1993.

\bibitem[BCWdW01]{buhrman2001quantum}
Harry Buhrman, Richard Cleve, John Watrous, and Ronald {d}e Wolf.
\newblock Quantum fingerprinting.
\newblock {\em Physical Review Letters}, 87(16):167902, 2001.

\bibitem[BW92]{bennett1992communication}
Charles~H. Bennett and Stephen~J. Wiesner.
\newblock Communication via one- and two-particle operators on
  einstein-podolsky-rosen states.
\newblock {\em Physical review letters}, 69(20):2881, 1992.

\bibitem[BZ08]{BZ08-covariance-matrices}
Zhidong Bai and Wang Zhou.
\newblock Large sample covariance matrices without independence structures in
  columns.
\newblock {\em Statistica Sinica}, 18(2):425--442, 2008.

\bibitem[CGJV19]{CGJV19-verifiable-qc}
Andrea Coladangelo, Alex~B. Grilo, Stacey Jeffery, and Thomas Vidick.
\newblock Verifier-on-a-leash: New schemes for verifiable delegated quantum
  computation, with quasilinear resources.
\newblock In Yuval Ishai and Vincent Rijmen, editors, {\em Advances in
  Cryptology -- EUROCRYPT 2019}, volume 11478 of {\em Lecture Notes in Computer
  Science}, pages 247--277, Cham, 2019. Springer International Publishing.

\bibitem[Cha20]{Charezma}
Michael Charezma.
\newblock Quantum circuit diagrams.
\newblock
  \url{https://warwick.ac.uk/fac/sci/physics/research/cfsa/people/pastmembers/charemzam/pastprojects},
  2006 (accessed October 14, 2020).

\bibitem[Cur79]{Curlander79-distinguishability}
Paul~Joseph Curlander.
\newblock {\em Quantum Limitations on Communication Systems}.
\newblock PhD thesis, Massachusetts Institute of Technology, Dept. of
  Electrical Engineering and Computer Science, 1979.

\bibitem[FK19]{farkas2019self}
M{\'a}t{\'e} Farkas and J{\k{e}}drzej Kaniewski.
\newblock Self-testing mutually unbiased bases in the prepare-and-measure
  scenario.
\newblock {\em Physical Review A}, 99(3):032316, 2019.

\bibitem[FKN22]{FKN22-mum-superdense-coding}
M\'{a}t\'{e} Farkas, J\k{e}drzej Kaniewski, and Ashwin Nayak.
\newblock Mutually unbiased measurements, {Hadamard} matrices, and {Superdense
  Coding}.
\newblock Technical Report arXiv:2204.11886 [quant-ph], ArXiv.org Preprint
  Archive, \texttt{https://www.arxiv.org/}, April 2022.

\bibitem[Hal35]{Hall35-matching}
Philip Hall.
\newblock On representatives of subsets.
\newblock {\em Journal of the London Mathematical Society}, s1-10(1):26--30,
  1935.

\bibitem[HJW93]{hughston1993complete}
Lane~P. Hughston, Richard Jozsa, and William~K. Wootters.
\newblock A complete classification of quantum ensembles having a given density
  matrix.
\newblock {\em Physics Letters A}, 183(1):14--18, 1993.

\bibitem[JNV{\etalchar{+}}20]{ji2020mip}
Zhengfeng Ji, Anand Natarajan, Thomas Vidick, John Wright, and Henry Yuen.
\newblock $\mathsf{MIP}^* = \mathsf{RE}$.
\newblock Technical Report arXiv:2001.04383 [quant-ph], ArXiv.org Preprint
  Archive, \texttt{https://www.arxiv.org/}, January 2020.

\bibitem[Kho79]{Holevo79-distinguishability}
Alexander~S. Kholevo.
\newblock On asymptotically optimal hypothesis testing in quantum statistics.
\newblock {\em Theory of Probability \& Its Applications}, 23(2):411--415,
  1979.

\bibitem[KR03]{KR03-unitary-error-bases}
Andreas Klappenecker and Martin R{\"o}tteler.
\newblock Unitary error bases: Constructions, equivalence, and applications.
\newblock In Marc Fossorier, Tom H{\o}holdt, and Alain Poli, editors, {\em
  Proceedings of the 15th International Symposium on Applied Algebra, Algebraic
  Algorithms, and Error-Correcting Codes (AAECC)}, volume 2643 of {\em Lecture
  Notes in Computer Science}, pages 139--149. Springer, Berlin / Heidelberg,
  Germany, May12--16, 2003.

\bibitem[Mec19]{M19-random-matrix-theory}
Elizabeth~S. Meckes.
\newblock {\em The Random Matrix Theory of the Classical Compact Groups},
  volume 218 of {\em Cambridge Tracts in Mathematics}.
\newblock Cambridge University Press, July 2019.

\bibitem[Mon07]{montanaro2007distinguishability}
Ashley Montanaro.
\newblock On the distinguishability of random quantum states.
\newblock {\em Communications in Mathematical Physics}, 273(3):619--636, August
  2007.

\bibitem[MV16]{MV16-unitary-error-bases}
Benjamin Musto and Jamie Vicary.
\newblock Quantum {Latin} squares and unitary error bases.
\newblock {\em Quantum Information and Computation}, 16(15-16):1318--1332,
  November 2016.

\bibitem[MY98]{mayers1998quantum}
Dominic Mayers and Andrew Yao.
\newblock Quantum cryptography with imperfect apparatus.
\newblock In {\em Proceedings 39th Annual Symposium on Foundations of Computer
  Science (Cat. No. 98CB36280)}, pages 503--509. IEEE, 1998.

\bibitem[MY04]{mayers2003self}
Dominic Mayers and Andrew Yao.
\newblock Self testing quantum apparatus.
\newblock {\em Quantum Information \& Computation}, 4(4):273--286, July 2004.

\bibitem[NC11]{NC11-quantum-information}
Michael~A. Nielsen and Isaac~L. Chuang.
\newblock {\em Quantum Computation and Quantum Information}.
\newblock Cambridge University Press, New York, NY, USA, 2011.
\newblock 10th Anniversary Edition.

\bibitem[NY20]{NY20-rigidity}
Ashwin Nayak and Henry Yuen.
\newblock Rigidity of superdense coding.
\newblock Technical Report arXiv:2012.01672v1 [quant-ph], arXiv Pre-print
  server, \texttt{https://arxiv.org/abs/2012.01672}, December 2020.

\bibitem[ON99]{ON99-converse-channel-coding}
Tomohiro Ogawa and Hiroshi Nagaoka.
\newblock Strong converse to the quantum channel coding theorem.
\newblock {\em IEEE Transactions on Information Theory}, 45(7):2486--2489,
  1999.

\bibitem[{\v{S}}B20]{vsupic2020self}
Ivan {\v{S}}upi{\'c} and Joseph Bowles.
\newblock Self-testing of quantum systems: a review.
\newblock {\em Quantum}, 4:337, 2020.

\bibitem[Sch35]{schrodinger1935discussion}
Erwin Schr{\"o}dinger.
\newblock Discussion of probability relations between separated systems.
\newblock {\em Mathematical Proceedings of the Cambridge Philosophical
  Society}, 31(4):555--563, 1935.

\bibitem[TKV{\etalchar{+}}18]{tavakoli2018self}
Armin Tavakoli, J{\k{e}}drzej Kaniewski, Tam{\'a}s V{\'e}rtesi, Denis Rosset,
  and Nicolas Brunner.
\newblock Self-testing quantum states and measurements in the
  prepare-and-measure scenario.
\newblock {\em Physical Review A}, 98(6):062307, 2018.

\bibitem[Tys09a]{Tyson09-erratum}
Jon Tyson.
\newblock Erratum: ``{Minimum}-error quantum distinguishability bounds from
  matrix monotone functions: {A} comment on `{Two}-sided estimates of
  minimum-error distinguishability of mixed quantum states via generalized
  {Holevo-Curlander} bounds'\,'' [{J. Math. Phys.} 50, 062102 (2009)].
\newblock {\em Journal of Mathematical Physics}, 50(10):109902, 2009.

\bibitem[Tys09b]{Tyson09-distinguishability}
Jon Tyson.
\newblock Two-sided estimates of minimum-error distinguishability of mixed
  quantum states via generalized {Holevo-Curlander} bounds.
\newblock {\em Journal of Mathematical Physics}, 50(3):032106, 2009.

\bibitem[Tys10]{Tyson10-error-recovery}
Jon Tyson.
\newblock Two-sided bounds on minimum-error quantum measurement, on the
  reversibility of quantum dynamics, and on maximum overlap using directional
  iterates.
\newblock {\em Journal of Mathematical Physics}, 51(9):092204, 2010.

\bibitem[Uhl76]{uhlmann1976transition}
Armin Uhlmann.
\newblock The ``transition probability'' in the state space of a $*$-algebra.
\newblock {\em Reports on Mathematical Physics}, 9(2):273--279, 1976.

\bibitem[Ver18]{Vershynin18-high-dim-probability}
Roman Vershynin.
\newblock {\em High-Dimensional Probability: An Introduction with Applications
  in Data Science}, volume~47 of {\em Cambridge Series in Statistical and
  Probabilistic Mathematics}.
\newblock Cambridge University Press, Cambridge, UK, 2018.

\bibitem[VV19]{vazirani2019fully}
Umesh Vazirani and Thomas Vidick.
\newblock Fully device independent quantum key distribution.
\newblock {\em Communications of the ACM}, 62(4):133--133, 2019.

\bibitem[VW00]{VW00-Pauli-operators}
Karl Gerd~H. Vollbrecht and Reinhard~F. Werner.
\newblock Why two qubits are special.
\newblock {\em Journal of Mathematical Physics}, 41(10):6772--6782, 2000.

\bibitem[Wat18]{Watrous18-TQI}
John Watrous.
\newblock {\em The Theory of Quantum Information}.
\newblock Cambridge University Press, May 2018.

\bibitem[Wer01]{Werner01-teleportation-dense-coding}
Reinhard~F. Werner.
\newblock All teleportation and dense coding schemes.
\newblock {\em Journal of Physics A: Mathematical and General},
  34(35):7081--7094, August 2001.

\bibitem[Wil13]{wilde2013quantum}
Mark~M. Wilde.
\newblock {\em Quantum Information Theory}.
\newblock Cambridge University Press, Cambridge, UK, 2013.

\bibitem[Yas16]{Yaskov16-MP-short-proof}
Pavel Yaskov.
\newblock A short proof of the {Marchenko–Pastur} theorem. {Une} courte
  d{\'e}monstration du th{\'e}or{\`e}me de {Marchenko–Pastur}.
\newblock {\em Comptes Rendus Math{\'e}matique}, 354(3):319--322, March 2016.

\end{thebibliography}

\end{document}